\documentclass[11pt]{article}

\usepackage[margin=1in]{geometry}
\usepackage{microtype}
\usepackage{graphicx}
\usepackage[colorlinks=true,linkcolor=blue,citecolor=blue,urlcolor=blue]{hyperref}

\usepackage{amsmath,amssymb,amsthm,amsfonts,bm,mathrsfs,braket}
\usepackage{cleveref}

\usepackage[most]{tcolorbox}

\newtheorem{theorem}{Theorem}
\newtheorem{lemma}{Lemma}
\newtheorem{corollary}{Corollary}
\newtheorem{definition}{Definition}
\newtheorem{conjecture}{Conjecture}

\newtheorem{procedure}{Procedure}

\usepackage{algorithm}
\usepackage[noend]{algpseudocode}
\algrenewcommand{\algorithmicensure}{\textbf{Output:}}

\newcommand{\eps}{\epsilon}
\newcommand{\R}{\mathbb{R}}
\newcommand{\calT}{\mathcal{T}}
\newcommand{\calP}{\mathcal{P}}
\newcommand{\calF}{\mathcal{F}}
\newcommand{\calA}{\mathcal{A}}
\newcommand{\calE}{\mathcal{E}}
\newcommand{\calH}{\mathcal{H}}
\newcommand{\calN}{\mathcal{N}}
\newcommand{\calB}{\mathcal{B}}
\newcommand{\calS}{\mathcal{S}}
\newcommand{\x}{\mathbf{x}}
\newcommand{\rr}{\mathbf{r}}
\newcommand{\g}{\mathbf{g}}
\newcommand{\ii}{\mathbf{i}}
\newcommand{\poly}{\mathrm{poly}}
\newcommand{\Tr}{\operatorname{Tr}}
\newcommand{\ketbra}[2]{\ket{#1}\!\bra{#2}}
\usepackage{xcolor}

\title{Quantum advantage from random geometrically-two-local Hamiltonian dynamics}

\author{Yihui Quek \thanks{Massachusetts Institute of Technology.} \thanks{École Polytechnique Fédérale de Lausanne.}}
\date{October 5, 2025}

\begin{document}
\maketitle
\setcounter{footnote}{0}

\begin{abstract}
Classical hardness-of-sampling results are largely established for random quantum circuits, whereas analog simulators natively realize time evolutions under geometrically local Hamiltonians. Does a {\em typical} such Hamiltonian already yield classically-intractable dynamics? We answer this question in the affirmative for the ensemble of geometrically-2-local Hamiltonians with Gaussian coefficients, evolved for for constant time. This naturally leads to a quantum advantage scheme with clear prospects for experimental realization, necessitating only course-grained control. 

We give strong evidence of hardness for this physically-relevant ensemble. We develop the first worst-to-average-case reduction for approximating output probabilities of (time-independent) geometrically-2-local Hamiltonian evolutions. Our reduction proceeds by nonstandard means: while we also leverage polynomial interpolation, unlike previous works, we reduce to an evaluator for the {\em exact} distribution over Hamiltonians from which we are trying to prove that sampling is hard. Previous works instead sampled from various perturbations of this distribution, introducing additional constraints meant to keep the perturbation, measured in total variation distance, under control. We dispense with this step.

Our reduction consists in a robust multivariate polynomial interpolation (where the polynomial is in the coefficients of the Gaussian), reduced to sequential robust univariate interpolations via the symmetries of the Gaussian. We circumvent the fact that random Hamiltonians lack a {\em hiding} symmetry, a key property in previous proofs. We also contribute an {\em algorithmic} version of Berlekamp-Welch to deal with errored evaluations, solving an open problem from the RCS literature. This strengthens the evidence for hardness of approximating probabilities offered in prior works: should a classical evaluator exist, the level of the Polynomial Hierarchy ($\textsf{PH}$) to which we can force a collapse tightens by one oracle level. We expect the machinery we develop to find use in average-case Hamiltonian complexity, filling in a gap in this literature which has thus far focussed on worst-case hardness results.
\end{abstract}

\newpage
\section{Introduction}
A central goal in near‑term quantum computing is to demonstrate quantum advantage. The hallmark proposals in this regard center around {\em sampling} tasks that (i) are naturally realized by existing experimental platforms, and (ii) are classically intractable under standard complexity assumptions. Random‑circuit sampling (RCS) \cite{BFNV18,BFLL21} has provided a crisp blueprint on the circuit side, with BosonSampling \cite{AA11, BDFH24, bouland2025complexitytheoreticfoundationsbosonsamplinglinear} and Fermion Sampling \cite{Oszmaniec_2022} being the leading candidates for photonic and fermionic linear optics systems. While many of these proposals have already been implemented in experiment \cite{Arute2019QuantumSupremacy, boson_exp_1, boson_exp_2, boson_exp_3, boson_exp_4}, there is often a gap between the settings that can actually be realized in laboratories, and those for which classical hardness can be convincingly argued.

To convince skeptics, we must close such gaps between the experiment and theory of these quantum advantage demonstrations. Motivated by this, we give strong evidence that quantum advantage is attainable with only coarse-grained experimental control. The overwhelming majority of sampling-based quantum-advantage demonstrations have come from {\em digitally-programmed} experimental setups, where one compiles a desired circuit into a sequence of calibrated one‑ and two‑qubit gates and measures at the end. This compilation requires the use of fairly involved gadgets on certain experimental platforms (see for instance \cite{Cesa}), which naturally also accumulates noise. Our goal is to show that quantum advantage is not the sole domain of these precisely-tuned digital setups, but is availed by {\em most} Hamiltonians, with minimal  overhead. 

This complementary route to quantum advantage is as follows: prepare a product state, time evolve with a typical geometrically‑two-local Hamiltonian for constant time, and measure. This “quench‑and‑measure” task aligns with capabilities of Rydberg/neutral-atom arrays and optical‑lattice simulators. Previously, Refs. \cite{Bergamaschi_2024, rajakumar2024gibbssamplinggivesquantum} also proposed quantum advantage demonstrations via Gibbs sampling from geometrically-local Hamiltonians at temperature $\beta = \Theta(1)$; Ref. \cite{rajakumar2024gibbssamplinggivesquantum} in particular was able to get the required Hamiltonian locality down to five, on a 3D lattice. Ref. \cite{BermejoVega18} also proposed time-independent evolution of an Ising ($Z$-type) Hamiltonian on a random input state as a route to quantum advantage. However, to our knowledge we are the first to propose time-independent evolution of {\em random} two-local Hamiltonians -- commuting or not -- on a 2D lattice.

We, furthermore, give strong complexity-theoretic evidence of the hardness of our task. 
Our proof will never need to translate Hamiltonian time evolution into an equivalent circuit -- indeed, doing so via a Trotterization argument incurs additional losses and meets barriers that we detail in \Cref{subsec:notsuffice} (``What does {\em not} suffice to prove hardness"). Instead we develop a number of new tools for worst-to-average-case reductions within the ensemble of random Hamiltonians itself, that we expect to find use in future investigations of of average-case Hamiltonian complexity.

\medskip
\textbf{Our contributions.} We introduce the following Hamiltonian ensemble on a fixed-degree lattice: for the set $\calP_k$ of geometrically $k$-local unsigned Pauli terms on that lattice, draw independent coefficients from a Gaussian and define the ensemble $\mathcal{E}(k)$ as
$$
H(\g)= \frac{1}{\sqrt{|\calP_k|}}\sum_{P \in \calP_k} g_P P, \qquad g_P \stackrel{\text{i.i.d.}}{\sim} \mathcal N(0,1).
$$
We call this Hamiltonian ensemble the {\em Gaussian geolocal Hamiltonian ensemble}. The Gaussian geolocal Hamiltonian ensemble is physically relevant because of its geometric locality; for this reason it is also natural to implement for analog simulators. In the rest of this work, for simplicity, we will focus on $k=2$ and $2D$ lattices, but our proof can easily be extended to larger $k$s and arbitrary graph structures -- as long as a worst-case family of instances exists on the same graph.  

We give strong complexity-theoretic evidence that it is classically hard to sample from the output distribution of the following experiment, which we term {\em Random Hamiltonian Sampling}: 
\begin{enumerate}
    \item Arrange $n = m\times m'$ qubits side-by-side on an $m$ row, $n$ column square lattice, and prepare the product-state input $|+\rangle^{\otimes n}$.
\item Let the system time-evolve under a random $H(\g)$ for a fixed evolution time $\tau = 1$. 
    \item Measure all qubits in the $X$ basis.
\end{enumerate}

\medskip
\textbf{Contribution (I):} To argue that the above experiment is classically hard, our main contribution is a {\em worst-to-average-case reduction} for the task of computing output probabilities of random Hamiltonian time evolutions. Intuitively, such a reduction implies that ``most" Hamiltonians in $\calE(2)$ are just as hard as the worst-case instance. We convey the gist of how such a reduction works, starting from the familiar setting of a quantum advantage proposals based on sampling from random {\em circuits}. First, one approximates the desired output probability by a polynomial in the gates of the circuit. Because of this polynomial approximation, if an efficient algorithm $\calA$ (an ``average-case" evaluator) can evaluate those probabilities at {\em most} given input points, it is also implicitly evaluating this polynomial. Since the target output probability arising from a worst-case instance is also well-approximated by this polynomial at a particular `worst-case' value, one could obtain the target value by using $\calA$'s evaluations on randomly-sampled instances (circuits with random gates) as interpolation points, and interpolate this polynomial to the worst-case value. All-in-all, we have an algorithm that solves any worst-case instance with high probability, completing the reduction. 

Let us describe how this general blueprint plays out in our setting. For the worst-case instance, we adapt the construction of Ref. \cite{BermejoVega18} which concerns the task of constant-time Ising time evolution of a {\em random} product input state. For a worst-case member of this ensemble of time evolutions, evaluation of the output probabilities to additive error $2^{-{O(n)}}$ is $\# \mathsf{P}$ hard. The complexity class $\# \textsf{P}$ is the class of counting problems, which are believed to be exceedingly difficult -- well beyond classical probabilistic computation ($\textsf{BPP}$) and even beyond $\textsf{NP}$ and $\textsf{PH}$ (the Polynomial Hierarchy).

Our worst-to-average-case reduction then says that if one could evaluate the output probability of a typical (i.e. average-case) Hamiltonian from $\calE(2)$, one could also do so for {\em any} (i.e. worst-case) Hamiltonian in the support of $\calE(2)$ -- including the ones in the previous paragraph. To our knowledge, all existing such reductions require the average-case solver to succeed on a set of instances sampled from a {\em deformation} of the average-case ensemble. This can pose two problems: first, the perturbed instances may even leave the class of evolutions under consideration (as is the case for IQP circuits); second, even if they remain within the class of evolutions, one has to argue about TVD closeness of the deformed distribution and the true distribution. 

By contrast, our technique cleanly reduces worst-case evaluation to an evaluator that succeeds on the {\em exact} average-case distribution, so we do not have to deal with either of these artifacts. Since we approximate the output probability with a multivariate polynomial in the coefficients of the Hamiltonian, our reduction is ultimately a {\em multivariate} polynomial interpolation in $\g$. We proceed by reducing this multivariate problem to a series of univariate polynomial interpolations, leveraging the spherical symmetry of the $l$-variate Gaussian to identify lower-dimensional submanifolds of $\mathbb{R}^l$ to interpolate on. We call this technique `slicing and dicing the sphere', and anticipate that it may also find use in other average-case Hamiltonian complexity problems where polynomial approximation plays a central role. 

\medskip
\textbf{Contribution (II):} The reason we focus on the task of computing probabilities is that, via Stockmeyer's reduction, this ultimately connects up with the task we wish to rule out: sampling from the output distribution of a random Hamiltonian time evolution, for a Hamiltonian sampled from $\calE(2).$ This reduction says: If there were indeed an efficient sampler $\calS$ over $\calE(2)$, there would also exist an algorithm to compute probabilities in $\textsf{BPP}^{\textsf{NP}^\calS}$ for an average $H\sim \calE(2)$. While this overall line of attack is familiar in the quantum advantage literature, one of our technical contributions is pushing it through for the specific class of evolutions we are dealing with. The difficulty arises because our ensemble of random Hamiltonians lacks a certain ``hiding" property used in previous hardness proofs. This property, were it to hold, would say that the distribution over Hamiltonian evolutions induced by picking $H\sim \calE(2)$ is {\em invariant} under appending a layer of $X$ gates at the output. Unfortunately, this is not the case. We resolve this problem by observing that jointly randomizing both the input state and the Hamiltonian gives us the distributional invariance we require.

\medskip
\textbf{Contribution (III):} Lastly, we point out a technical aspect of our proof that may be of independent interest. We show that average-case probability evaluation is in $\textsf{BPP}^{\calS}$, while previous works \cite{BFNV18,BFLL21,BDFH24,Oszmaniec_2022} were only able to prove that average-case probability evaluation for the qubit, bosonic and fermionic circuit analogs of our sampling task is in $\textsf{BPP}^{\textsf{NP}^\calS}$, one level above in the polynomial hierarchy. Our result similarly strengthens all of these precedents to $\textsf{BPP}^{\calS}$, positively answering an open question in \cite{BFLL21}. 

Our improvement stems from an algorithmic version of the Berlekamp-Welch algorithm that we develop. The Berlekamp-Welch algorithm \cite{WB86}, which was developed in the context of error correction (decoding Reed-Solomon codes), is a way to perform polynomial interpolation while correcting for a portion of errored interpolation points. In the worst-to-average-case reduction, it arises as part of the polynomial interpolation step. We develop an extension of this algorithm that is robust to both a constant fraction of the points being corrupted, as well as small additive errors in the uncorrupted points. While the original Berlekamp-Welch solves a set of linear equations to find an ``error-locator" polynomial that vanishes exactly at corrupted positions, our robust version of Berlekamp-Welch solves two linear programs: the first to build a selection mask $s(\g)$ whose zeros align with the corrupted points, and the second to recover the original polynomial by ``dividing out" the mask. 


\medskip
\textbf{Previous work on average-case hardness of sampling from Hamiltonian dynamics.} While there has been considerable complexity-theoretic work done on the average-case hardness of random circuit sampling \cite{BFNV18,kondo2021quantumsupremacyhardnessestimating,BFLL21}, BosonSampling \cite{bouland2025complexitytheoreticfoundationsbosonsamplinglinear, BDFH24} and even FermionSampling \cite{Oszmaniec_2022}, the hardness of Hamiltonian dynamics is much less well-studied. 

Ref.\cite{Haferkamp20}, which also set out to devise a quantum advantage scheme feasible on analog quantum simulators, provided evidence that a {\em fixed}, Ising Hamiltonian initialized on a random product state, is hard to sample from. Our setting differs from theirs in that we consider the time evolution of a {\em random} Hamiltonian on a fixed input state. Moreover, they show average-case hardness for {\em exactly} evaluating the output probabilities of their architectures, while we work in the more physically realistic setting that tolerates a small imprecision in the evaluations. 

Ref. \cite{ParkXanadu} is closest to us in setting. They also consider sampling random Hamiltonian dynamics for a very similar Hamiltonian ensemble, but are only able to prove hardness of probability evaluation of high-Hamming weight output bitstrings, up to an additive error of $n^{-\Theta(n^3 \log n)}$ (Theorem 1). Moreover, they require a conjecture about anticoncentration and their Hamiltonians are restricted to being on a bipartite graph. 
We prove hardness up to a much more permissive additive error of $n^{-\Theta(n)}$, and our setting does not have these restrictions.


\section{Our results}
For pedagogical reasons, we present our main results in a different order from the introduction. We kick off this section by providing some crucial context for our main result: in \Cref{subsec:complexitywithouthiding}, we give an overview of the complexity-theoretic machinery to establish that sampling from Hamiltonian dynamics is hard. While the overarching framework follows that of previous works, we first have to construct a worst-case-hard family of Hamiltonians. We also devote significant effort to resolving an obstacle that is unique to our setting: random Hamiltonians lack a certain {\em hiding} symmetry which was used in previous proofs.

Next, in \Cref{subsec:worst-av-casereduction}, we dive into one step of this machinery, wherein also lies our main contribution: the worst-to-average-case reduction for the problem of evaluating output probabilities. \Cref{subsec:algorithmicRBW} further zooms into one particular step of this worst-to-average-case reduction, contributing an algorithmic version of the robust Berlekamp-Welch algorithm, which performs polynomial interpolation in the presence of a small fraction of errored evaluations. Previous works were only able to establish that this task is possible given access tp an \textsf{NP} oracle, and left the existence of an algorithmic version open. Our algorithm thus resolves this open question, thereby strengthening the complexity-theoretic evidence for the hardness of evaluating probabilities. 



\subsection{Complexity wihout hiding}\label{subsec:complexitywithouthiding}
The remainder of this paper is focused on proving classical hardness for evaluating probabilities. Yet, our quantum advantage proposal concerns a different task: {\em sampling} from the output of a random Hamiltonian evolution. How are the two tasks related?

We devote this subsection to providing context for our main contribution. We connect evaluating probabilities to sampling, and explain how these tasks fit into the evidence for quantum advantage. We will explain the argument at a high level in this subsection, deferring detailed proofs to \Cref{sec:overview}. Note that the arguments for hardness of random circuit sampling \cite{BFNV18}, boson sampling \cite{AA11, BDFH24} and fermion sampling \cite{Oszmaniec_2022} are all based off this template. However, a challenge unique to our ensemble of Hamiltonians is the lack of a certain invariance property known as {\em hiding}. Broadly, the argument shows that the assumed existence of an efficient classical sampler would lead to a contradiction of a well-believed complexity theory assumption. We present the argument in two parts (only in the second part do we assume that an efficient classical sampler exists):  

\textbf{(1) From worst– to average-case hardness of evaluating a single output probability.}
First, we construct a family of Hamiltonians for which the probability of outputting $\ket{+^n}$ is worst-case hard to evaluate. This construction is an adaptation of one in \cite{BermejoVega18}, which used a fixed Ising Hamiltonian to evolve an ensemble of input states. We instead fix the input state and define a worst–case hard family of Hamiltonians $\mathcal{E}_{worst}$, which are all geometrically-two-local Ising-type Hamiltonians. For $H_{\mathcal S}\in\mathcal{H}_{worst}$ and evolution time $\tau = O(1)$, the output probability
\[
D_{H_{\mathcal S}}:=\big|\!\braket{+^n|e^{-iH_{\mathcal S}\tau}|+^n}\!\big|^2
\]
is $\# \textsf{P}$–hard to approximate to additive error $2^{-\Theta(n)}$ in the worst case (Theorem~\ref{thm:worst}). Intuitively, the hardness arises from the fact that these constant–depth evolutions can encode amplitudes of arbitrary polynomial-time quantum computations, and so even very fine additive estimates of those amplitudes would solve $\# \textsf{P}$–hard problems.

Then, our main contribution is a worst–to–average–case reduction (\Cref{thm:mainthm_informal}). This reduces the task of evaluating output probabilities arising from any Hamiltonian in the above worst-case ensemble, $\mathcal{E}_{worst}$, to evaluating those of a {\em typical} Hamiltonian in our Gaussian geo–local ensemble $\mathcal{E}(2)$. This yields \emph{average–case} $\# \textsf{P}$–hardness of probability evaluation over $\mathcal{E}(2)$. We will actually use a slightly stronger form of average-case hardness, where the algorithm has to succeed over a joint randomization of both the Hamiltonian and the input state -- this is for reasons we will explain shortly.

\textbf{(2) From sampling to evaluation in $\textsf{BPP}^{\textsf{NP}^{\calS}}$.}
Assume for the sake of contradiction that an \emph{efficient classical sampler} $\calS$ exists: a polynomial–time algorithm that, for most instances in $\calE(2)$, produces samples within total–variation distance $\varepsilon$ of the true distribution. By Stockmeyer’s approximate counting \cite{Stockmeyer85}, such a sampler implies an {\em evaluator}, which approximates output probabilities: an algorithm that gives \emph{additive} $O(2^{-n})$ estimates of the true probabilities for most outcomes $x$ (Lemma~\ref{lem:samplerimpliesevaluator}). This evaluator is in $\textsf{BPP}^{\textsf{NP}^\calS}$ and it succeeds on average {\em over outcomes}. This does not quite close the loop and complete the contradiction, however: notice that Part (1) concerns an evaluator that succeeds on average {\em over instances} -- not outcomes. This is the last step we will now tackle. 

For quantum advantage proposals based on sampling from random circuits and fermionic linear optics circuits, one converts “success over most \emph{outcomes} (for a fixed circuit)” into “success over most \emph{circuits} (for a fixed outcome)” using a \emph{hiding} symmetry: appending random $X$’s relabels the output string, without changing the circuit distribution. Unfortunately, this symmetry \emph{fails} for our Hamiltonian ensemble $\calE(2)$: 
\[
H\sim\mathcal{E}(2)\quad\not\Longrightarrow\quad e^{-iH\tau}Z^y\ \stackrel{d}{=}\ e^{-iH\tau}\quad\text{for any }y,
\]
where $\stackrel{d}{=}$ means equality in the sense of distributions.

\medskip\noindent\textbf{Overcoming the lack of “hiding” for Hamiltonians.}
Our workaround is to \emph{randomize the input} as well. We use a simple invariance of the ensemble:
\[
H\sim\mathcal{E}(2)\quad\Longrightarrow\quad Z^yHZ^y\ \stackrel{d}{=}\ H\quad\text{for any }y.
\]
This invariance allows us to `push' the randomizing $Z^y$ string, initially at the output, into the exponent of the time evolution, whereupon it emerges on ``the other side", at once randomizing both the Hamiltonian and the input:
\begin{equation}
\bra{+^n}Z^y e^{-iH\tau}\ket{+^n} = \bra{+^n}e^{-i(Z^yHZ^y)\tau} Z^y\ket{+^n}.
\end{equation}
The upshot is that we obtain a \emph{joint hiding} property (Theorem~\ref{thm:hiding}), which holds {\em jointly} over the input state and the distribution over Hamiltonian instances. Using this, we successfully map ``success over most outcomes” to ``success over most random input-instance pairs” (Theorem~\ref{thm:prfixedoutput}), completing the sampling $\Rightarrow$ evaluation step for random Hamiltonians.

\medskip
Putting the pieces together gives the  contradiction:

\begin{center}
    {\em Average–case approximate Hamiltonian evaluation is $\# \textsf{P}$–hard, yet any polynomial–time classical sampler $\calS$ that is $\varepsilon$–close in total variation on most instances would place it in $\textsf{BPP}^{\textsf{NP}^\calS}$.}
\end{center}

By Toda’s theorem \cite{Toda91}, this would collapse \textsf{PH} to the third level; under the standard assumption that \textsf{PH} does not collapse, no such sampler exists.

\subsection{Worst-to-average-case reduction}\label{subsec:worst-av-casereduction}

A core step of the above schema is the worst-to-average-case reduction for evaluating probabilities, used in Step (1). We prove:
\begin{theorem}[Worst-to-average-case reduction (\Cref{thm:outer}, adapted)]\label{thm:mainthm_informal}
The following task is $\# \mathsf{P}$-hard under a $\textsf{BPP}$ reduction: for any constant $\eta$ and for $\tau=1$, upon input a random Hamiltonian $H\sim \calE(2)$, output an additive approximation $\hat{p}$ to the output probability  $|\bra{+^n} e^{-iH\tau} \ket{+^{n}}|^2$ satisfying
    \begin{equation}
        |\hat{p}- |\bra{+^n} e^{-iH\tau} \ket{+^{n}}|^2| \leq 2^{-n \log(n)}
    \end{equation}
    with probability at least $1-\eta$ over the choice of $H$.
\end{theorem}
The full proof of this Theorem is deferred to \Cref{sec:worst-av}. We first assume that there exists an algorithm $\calA$ that can perform the above task. Then, we approximate $p_H$ with an $m$-th order Taylor polynomial in $\mathbf{g} \in \mathbb{R}^l$, the coefficients of the Hamiltonian $H$. Call the approximating polynomial $\mathcal{T}_m(\bf g)$; it suffices to take $m= \Theta(n)$. Let $\bf{g}_{\text worst}$ be the parameters corresponding to $H_{\text{worst}}.$ 

Unlike previous works, we will not be defining a single-parameter `interpolation path' that smoothly runs from a single, randomly-sampled, average-case instance to the worst-case instance. In random circuit sampling works, significant effort went towards ensuring that this path did not leave the unitary group (it would be strange to require the average-case evaluator to work on non-unitaries!). We never have this problem. Our main innovation is a way to cleanly construct a {\em multi-parameter interpolation path}  on the fly, based on samples from (marginals of) the {\em exact} average-case ensemble $\calE(2)$. 

We will argue that by sampling multiple Hamiltonians $\mathbf{g}_i \sim \mathcal{E}(2)$ and running $\calA$ on each of them, the (possibly errored) evaluations $\{\calA(\g_i) \approx \mathcal{T}_m(\g_i)\}_i$ can be used to interpolate to the value of $\mathcal{T}_m(\g_{\text worst})$. The Berlekamp-Welch algorithm \cite{WB86} was used in previous works for this polynomial interpolation task, and it is natural to try to use it for our setting. This plan immediately runs into some barriers: 
\begin{enumerate}
    \item \textbf{Multivariate polynomials:} Since the polynomial of interest, $\calT_m$, is a multivariate polynomial in $\g$ (the vector of Hamiltonian coefficients) we would like to perform multivariate, not univariate, polynomial approximation;
    \item \textbf{Allowing errors in evaluations:} When using Berlekamp-Welch on $\calA$'s evaluations, we would like this procedure to be robust to allowing $\mathcal{A}$ to make {\em some} approximation error, instead of assuming that it either fails outright, or outputs {\em exact} evaluations of $p_H$.
\end{enumerate}
The Berlekamp-Welch algorithm, as originally stated, performs univariate polynomial interpolation and requires at least a certain number of evaluations to be {\em exact}. It is not straightforward to generalize this to the multivariate case, as a cornerstone of the proof collapses: while a univariate polynomial can have a number of roots that is at most its degree, a multivariate polynomial of total degree $d > 0$ can have infinitely many roots. 

We solve both of these problems, showing that we can reduce 
multivariate polynomial interpolation into a sequence of robust univariate polynomial interpolations. The proof crucially leverages the spherical symmetry of the multivariate Gaussian. A key contribution is also an algorithmic version of a variant of the Berlekamp-Welch algorithm that is robust in the sense described above, solving an open problem from the RCS literature.



\subsubsection*{Slicing and dicing the sphere}
Even though the average-case evaluator is guaranteed to succeed on a random point from $\calN_l$, we may enjoy its success guarantees without necessarily sampling from all of $\calN_l$, but only from a well-chosen marginal. This follows from a counting argument (\Cref{lemma:2}), and a similar observation in the context of random circuits was already made as early as \cite{BFNV18}. 

A similar logic holds for random Hamiltonians. Analogously to random circuits, the average-case Hamiltonian evaluator will succeed with high probability when fed with samples from $\calN_l$'s marginal on a random plane. Let $\g_{\text{worst}}$ be the worst-case point: we will choose a random plane containing $\g_{\text{worst}}$ and the origin. Because the plane contains $\g_{\text{worst}}$, we may 
interpolate to $\g_{\text{worst}}$ based on the evaluations of $\calA$ at those points. This is already progress: since all these points lie in a plane, and any polynomial supported on a plane may be re-written as a bivariate polynomial, we have reduced our fearsomely multivariate task to that of bivariate interpolation.

We further reduce bivariate interpolation on a plane to a series of univariate interpolations, as follows: Without loss of generality, we may assume that $\g_{\text{worst}}$ is on the $z$-axis of the plane's coordinate system. If we could obtain evaluations at other points $\{\bf{r}_i\}$ that are also on the $z$-axis of $P_{\text{good}}$, we could then interpolate along the $z$-axis -- a univariate interpolation -- to obtain the value of $\g_{\text{worst}}$. How then could we obtain an evaluation at each $\bf{r}_i$ that we have high confidence in? We could not run $\calA$ directly on $\bf{r}_i$ as we do not have success guarantees for $\calA$ on the $z$-axis of $P_{\text{good}}$. The reason is that the $z$-axis of $P_{\text{good}}$ is a worst-case ray, as it is fully fixed from the moment $\g_{\text{worst}}$ is specified.

Instead, to obtain an estimate of $p_{\bf{r}_i}$, we will perform another univariate polynomial interpolation: We sample points from the circumference  that consists of all points in $P_{\text{good}}$ at distance $|\bf{r}_i|$. Call this circumference $C_i$. If the points are sampled from the marginal of $\calN_l$ on $C_i$, a counting argument guarantees that $\calA$ succeeds with high probability on most $C_i$'s when fed with points sampled in this way. Since $\bf{r}_i$ is also in $C_i$, we may then interpolate on this $C_i$ based on $\calA$'s evaluations. Naively this requires a bivariate polynomial interpolation in the variables $\cos(\theta), \sin(\theta)$ where $\theta$ is the angular coordinate of $C_i$. However we perform a re-parametrization to convert the problem into one of univariate interpolation. 
\begin{figure}[!htbp]
    \centering
    \includegraphics[width=0.5\linewidth]{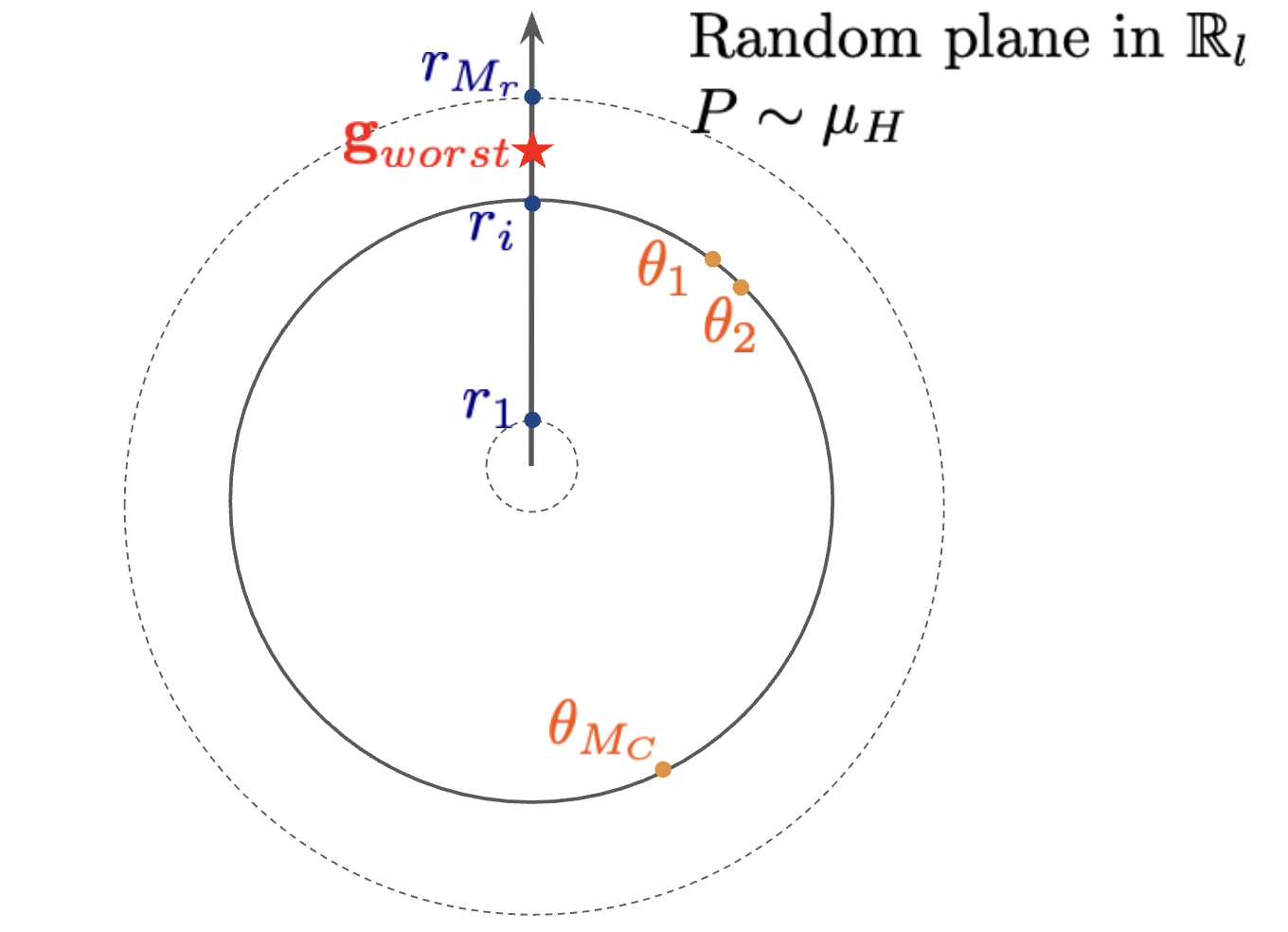}
    \caption{An illustration of steps 2 and 3 of our worst-to-average-case reduction, {\em slicing and dicing the sphere}. Having sampled a random plane (i.e. slice) $P$ containing the worst-case point $\g_{worst}$, the goal is now to interpolate to $\g_{worst}$ solely based on points in this slice.}
    \label{fig:slicingdicing_main}
\end{figure}

\medskip
All-in-all, the worst-to-average-case reduction consists in the following three steps, the last two of which are illustrated in \Cref{fig:slicingdicing_main}:
\begin{enumerate}
    \item \textbf{Sample a random plane $P$ containing $\g_{worst}$}.
    \item \textbf{Interpolation on the \(z\)-axis:} We interpolate a degree-$m$ polynomial on the points $\{(\rr_i, y_i)\}_i$  where $\rr_i\in P$ and is sampled from the radial distribution of $\calN_l$, and $y_i \approx \calT_m(\rr_i)$.
    \item \textbf{Interpolation on a circumference of radius \(|\rr_i|\):} To obtain $y_i$ at a given point $\rr_i$, we interpolate another degree-$m$ polynomial, $\calT_C$ on the points $\{(X_j^{(i)}, Y_j^{(i)})\}_j$ where the $X_j^{(i)}$ are sampled from the angular marginal of $\calN_l$ on the plane $P$, and exist on the circumference $C_i$ consisting of all the points on $P$ at distance $\rr_i$.
\end{enumerate}

\subsection{Algorithmic robust Berlekamp-Welch}\label{subsec:algorithmicRBW}
All in all, we have reduced multivariate interpolation on $\mathbb{R}^l$ to a polynomial number of univariate interpolations on a circumference and a univariate interpolation on a ray. Each of these univariate interpolations can be handled by the Berlekamp-Welch algorithm, of which we contribute a polynomial-time algorithmic version that is robust to two kinds of errors in its input points $(x_i, y_i)$: first, only a constant fraction of the $y_i$s are good approximations of $\calT_m(x_i)$; and secondly, the definition of a ``good" approximation is tolerant of some amount of imprecision in the evaluations. Previous work \cite{BFLL21} was only able to give a $\textsf{P}^{\textsf{NP}}$ algorithm for this problem, and left the elimination of the $\textsf{NP}$ oracle as an open problem. We solve this, giving a simple polynomial-time algorithm for the same problem and putting it in $\textsf{P}$.

\begin{theorem}[Robust efficient Berlekamp-Welch\label{thm:REBW}]
Let $x_1, \dots, x_n \in [-1,1]$ and let $y_1, \dots , y_n \in \R$.  Assume that $x_1,\dots , x_n$ are $\delta$-separated.  Let $0 <  k < n$ and assume that there exists a polynomial $p$ of degree at most $n - 2k - 1$ such that $|p(x_i) - y_i| \leq \eps$ for at least $n - k$ distinct indices $i$.  Then there is an algorithm that, on input $\{(x_i,y_i)\}_{i\in [n]}$, runs in time $\text{poly}(n)$ and outputs a polynomial $q$ of degree at most $n - 2k - 1$ such that  $|q(x_i) - p(x_i)| \leq (10/\delta)^{2n} \cdot \eps$ for at least $n-2k$ values $i \in [n]$.
\end{theorem}

The full proof of this theorem is deferred to \Cref{subsec:algorithmicRBW}. Let us jump ahead to its implications. The existence of a polynomial-time algorithm for robust Berlekamp-Welch implies that the task of evaluating output probabilities of random Hamiltonian time evolutions is $\# \mathsf{P}$-hard under $\textsf{BPP}$ reductions. Under the previous best algorithm for robust Berlekamp-Welch, one could at best say that the same task is $\# \mathsf{P}$-hard under $\textsf{BPP}^{\textsf{NP}}$ reductions. The former is stronger evidence for hardness: should probability evaluation turn out to be easy, the level of $\textsf{PH}$ to which we can force a collapse tightens by one oracle level. 

Since our algorithmic robust Berlekamp-Welch plugs directly into the complexity-theoretic arguments for the hardness of evaluating output probabilities of random circuits, Fermionic Linear Optics (FLO) circuits or permanents of BosonSampling setups, it also immediately implies a similar strengthening of the evidence for hardness for those tasks.

\begin{corollary}[$\# \textsf{P}$-hardness of SUPER in the saturated regime.] Let $\mid$ SUPER $\left.\right|_{ \pm} ^2$ be the problem of \textit{Sub-Unitary Permanent Estimation
with Repetitions}, defined in \cite{bouland2025complexitytheoreticfoundationsbosonsamplinglinear}. Let $m \geq 2.1 n$. In the regime $m=\Theta(n)|\mathrm{SUPER}|_{ \pm}^2$ is $\# \textsf{P}$-hard under $\textsf{BPP}$ reductions to additive error $\epsilon(S)=e^{-5 n \log (n)-O(n)}$. with probability at least $1-\delta$, with $\delta=1 / \operatorname{poly}((n)).$
\end{corollary}

\begin{corollary}[$\#\textsf{P}$ hardness of approximating random circuit output probabilities]
    Let $\mathcal{A}$ be a circuit architecture so that computing $\mathrm{p}_0(C) = |\bra{0}C\ket{0}|^2$ to within additive error $2^{-O(m)}$ is $\# \textsf{P}$-hard in the worst case. Then the following problem is $\# \textsf{P}$-hard under $\textsf{BPP}$ reductions: for any constant $\eta<\frac{1}{4}$, on input a random circuit $C \sim \mathcal{H}_{\mathcal{A}}$ with $m$ gates, compute the output probability $\mathrm{p}_0(C)$ up to additive error $\delta=\exp (-O(m \log m))$, with probability at least $1-\eta$ over the choice of $C$.
\end{corollary}

\section{Discussion}
\subsection{What does {\em not} suffice to prove hardness}\label{subsec:notsuffice}
The hardness of evaluating the output probabilities of random Hamiltonian dynamics {\em does not immediately follow} from the observation that every Hamiltonian time evolution can be approximated by a trotterized Hamiltonian time evolution, and thereafter seeking a recourse to random circuit sampling-style hardness tools. In general, this approach runs into the difficulty that we wish to rule out an evaluator that succeeds with high probability over the Gaussian geolocal distribution over Hamiltonians. Once we Trotterize this ensemble of random Hamiltonians, this induces a new distribution over the resulting circuits. This distribution over circuits is not natural in any sense, and in particular is not the same as the distribution over circuits obtained by sampling each gate from the Haar measure, which is the ensemble that the random circuit sampling proofs actually argue against. This could cause problems in using the hiding argument. In addition, to our knowledge, all existing proofs of hardness for evaluating the output probabilities of random circuits actually feed the evaluator inputs from a \textit{deformed }average-case distribution, whereupon they have to additionally argue that this deformed distribution is close to the Haar measure (see Fact 15 in \cite{BFNV18}, or Lemma 2 of \cite{movassagh2018efficientunitarypathsquantum}). Both of these arguments rely on properties of the Haar measure; it may be more difficult to make an analogous argument for the induced distribution over circuits obtained from Trotterization.

An alternative, ``RCS-style" tack for the worst-to-average-case reduction, would be to define a one-parameter path from a random coefficient vector $\g\sim \calE(2)$ to $\g_{worst}$ via
\begin{equation}
    \g_{\theta} = \theta \g + (1-\theta) \g_{worst}.
\end{equation}
Then, one would choose as interpolation points, the assumed evaluator's outputs when quizzed on the Hamiltonians $\{\g_{\theta_i}\}_i$ where $\theta_i\sim \text{Unif}([0,1])$. This is precisely the `deformed' average-case distribution alluded to many times in our text. 
This strategy, however, still encounters the second barrier above: to argue for TVD closeness of the deformed average-case distribution, one still has to find a way to relate the measurement output distribution for the evolution $e^{-iH(g_{\theta})t}$ to that for $e^{-iH(g_{0})t}$, which is arguably not simpler than the approach we have chosen. 

\subsection{Does anti-concentration hold?}
The reader may wonder if we can prove anti-concentration for our Hamiltonian ensemble, which is well-studied for random circuit sampling. Anti-concentration states that the output distribution of a random quantum circuit is ``spread out” – that most output probabilities are reasonably large. While not formally needed for our worst-to-average case proof, an anticoncentration statement would reduce our current notion of {\em additive} approximation that is necessary for the hardness proof to a more physically plausible one that involves only relative errors. 

We do not prove anticoncentration here. In fact, it is not clear to us if our ensemble of Hamiltonians time evolutions will anticoncentrate at all. We point out that random quantum circuits with two-local gates require $\Omega(\log(n))$ depth to anticoncentrate \cite{Dalzelletal}. Yet, it is debatable if this can be taken as heuristic evidence that our distributions (which result from running a Hamiltonian for constant time) do not anticoncentrate, as we cannot really equate `circuit depth' in that setting to runtime of our Hamiltonian time evolutions: runtime can always be rescaled by rescaling the distribution from which we sample the coefficients. 

Given that we have much fewer tools for analyzing random Hamiltonian time evolutions than for random circuits, we leave the question of whether our ensemble anticoncentrates as an intriguing open problem. 

\section{Preliminaries}
\subsection{The Gaussian geolocal Hamiltonian ensemble}\label{sec:preliminaries}
We introduce our geometrically $k$-local Hamiltonian ensemble. The notion of geometric locality relies on an underlying graph structure on which the qubits are arranged; we assume we have an $r$-dimensional lattice with $n$ qubits, but our results easily extend to Hamiltonians on arbitrary graphs. Let $\calP_k$ be the set of all geometrically-$k$-local Paulis on this lattice and set
$l := |\mathcal{P}_k|$; for fixed $r,k$ one has
$l = c_{r,k}\,n + O(1)$, i.e. $l=\Theta(n)$ with a constant
$c_{r,k}$ depending only on $r$ and $k$.

\begin{definition}[Gaussian geolocal Hamiltonian ensemble] With $l=|\calP_k|$, draw i.i.d. Gaussian coefficients
$\bm g=(g_P)_{P\in\mathcal{P}_k} \sim \mathcal N(0,\sigma^2 I_l)$
with $\sigma^2 = 1/l$ and define the following Hamiltonian ensemble, which we call $\calE(k)$:
\begin{equation}\label{eq:Hg}
  H(\bm g) \;:=\; \sum_{P\in\mathcal{P}_k} g_P\,P
  \;\stackrel{d}{=}\;
  \frac{1}{\sqrt{l}}\sum_{P\in\mathcal{P}_k} z_P\,P,
  \qquad z_P \stackrel{\text{i.i.d.}}{\sim} \mathcal N(0,1).
\end{equation}
\end{definition}

Often we will want to refer to the distribution from which the entire vector of coefficients $\g$ is sampled; we denote this as
\begin{equation}
    \calN_l := \mathcal{N}(0, \sigma^2 \mathbb{I}_{l}).
\end{equation}
where $\sigma^2 = \frac{1}{l}$. An alternative convention for the Hamiltonian's normalization could be to replace the $1/\sqrt{l}$ in \Cref{eq:Hg} with $1/\sqrt{n}$; this choice does not affect the asymptotics of our results because $l=\Theta(n)$. Alternatively, our entire proof machinery also goes through with the choice of normalization $1/n$, which would result in $\sigma^2=1/n^2$. 

In the rest of this paper, for simplicity, we will assume that the geometric locality is $k=2$ and that our Hamiltonians live on a rectangular 2D lattice of dimensions $n=L_x \times L_y$. Assuming periodic boundary conditions for simplicity, with this choice of parameters, 
\begin{equation}
    l = |\calP_2| \approx 3n+18n = 21n
\end{equation}
as there are 9 choices of 2-qubit Pauli terms per edge, for each of the $L_x(L_y-1)+L_y(L_x-1) \approx 2n$ edges, and 3 choices of 1-qubit Pauli terms per vertex, for each of the $L_xL_y =n$ vertices.

\subsection{Polynomial approximation of output probabilities}

We are interested in evaluating the probability of a sampled Hamiltonian time evolution outputting the bitstring $+^n$, i.e.
\begin{equation}
    D(\bm{g}) = |\bra{+^n} e^{-iH(\bm g)t} \ket{+^{n}}|^2,
\end{equation}
where $H(\g)\sim \calE(2).$ We will first approximate this probability by a polynomial in $\g$. 

Let us fix input and output states, calling them $\ket{\psi_{in}}$ and $\ket{\psi_{out}}$, a coefficient vector $\g$ and the corresponding Hamiltonian $H(\g)$ as in \Cref{eq:Hg}. The $m$-th order Taylor approximation of the Hamiltonian time evolution operator $e^{-iH(\g)t}$ is
\begin{align}
    \sum_{k=0}^{m} \frac{(-iH(\g)t)^k}{k!} = \sum_{k=0}^{m} \frac{(-it)^k}{k!} H(\g)^k = \sum_{k=0}^{m} c_k \sum_{\ii \in [|\calP_2|]^k} d_{\ii}^{(k)} S_{\ii}^{(k)},  
\end{align}
where we have assumed a canonical ordering of the Paulis in $\calP_2$ so that each can be referred to by an index in $\{1,\ldots |\calP_2|\} =: [|\calP_2|]$. Here 
\begin{align}
c_k &= \frac{-i^k t^k}{k!}\\
    d_{\ii}^{(k)} &= \prod_{j=1}^k g_{i_j} \qquad \text{(monomial in $\leq k$ variables of $\g$)}\\
    S_{\ii}^{(k)} &= P_{i_1}\ldots P_{i_k} \qquad \text{(ordered product of Paulis)}
\end{align}
We define the following $2m$-th order polynomial approximation of the output probability $D(\g)=|\bra{\psi_{out}} e^{-iH(\g)t}\ket{\psi_{out}}|^2$:
\begin{align}
    \mathcal{T}_{2m}(\g) &:= |\bra{\psi_{out}}\sum_{k=0}^{m} \frac{(-iH(\g)t)^k}{k!}\ket{\psi_{in}}|^2\\
    &=|\bra{\psi_{out}} \sum_{k=0}^{m} c_k \sum_{\ii \in [|\calP_2|]^k} d_{\ii}^{(k)} S_{\ii}^{(k)} \ket{\psi_{in}}|^2\\
    &= \sum_{k_1,k_2 =0}^m c_{k_1}c_{k_2}^{\ast}\sum_{\ii\in [|\mathcal{P}_2|]^{k_1}, \mathbf{j}\in [|\mathcal{P}_2|]^{k_2}} d_{\ii}^{(k_1)}d_{\mathbf{j}}^{(k_2)}\bra{\psi_{out}}S_{\ii}^{k_1} \ketbra{\psi_{in}}{\psi_{in}} S_{\mathbf{j}}^{k_2}  \ket{\psi_{out}}.\label{eq:polys}
\end{align}
Here, the only functions of $\g \in \mathbb{R}^l$ are the coefficients $d_{\ii}^{(k_1)},d_{\mathbf{j}}^{(k_2)}$ which is a monomial in at most $k_1+k_2$ variables. Therefore $\mathcal{T}_{2m}(\g)$ is a polynomial of degree-$2m$ in the vector $\g$. We note that the coefficients of $\mathcal{T}_{2m}$ will depend on $\ket{\psi_{in}}$, $\ket{\psi_{out}}$ and $t$, but we assume they are fixed. Let us now evaluate the quality of this polynomial approximation as a function of $m$.

\begin{theorem}[Polynomial approximation for $D$]\label{thm:polyapproxguarantees}
Let $\g \in \mathbb{R}^l$, and $H(\g):=\sum_{i\in \calP_2} g_i P_i$ where each $P_i$ is a unique $n$-qubit Pauli. Let $\bigl\{\mathcal{T}_{2m} : \mathbb{R}^l \!\to\! \mathbb{R}\bigr\}_{m\in\mathbb{N}}$
 be the family of degree $2m$ polynomial approximations given in \Cref{eq:polys}.We have
    \begin{equation}
    | D(\g) - \mathcal{T}_{2m}(\g) |\leq \eps
\end{equation}
as long as 
\begin{equation}\label{eq:polydegree}
    m \geq \Theta(\lVert H(\g)\rVert t + \log(1/\eps)).
\end{equation}

\end{theorem}

\begin{proof}
Let us start by bounding the operator norms of some relevant operators:
\begin{align}\label{eq:1}
    \left\lVert e^{-iH(\g)t} - \sum_{k=0}^{m} \frac{(-iH(\g)t)^k}{k!} \right\rVert  &=\left\|\sum_{k=m+1}^{\infty} \frac{(-i H(\g) t)^k}{k!}\right\| \\
    &\leq \sum_{k=m+1}^{\infty} \frac{\left\lVert H(\g) \right\rVert^k t^k}{k!} \\
    &\leq \sum_{k=m+1}^{\infty} \left(\frac{e\lVert H(\g) \rVert t}{k}\right)^k,\\
    &\leq \sum_{k=m+1}^{\infty} \left(\frac{e\lVert H(\g) \rVert t}{m}\right)^k 
\end{align}
where in the second inequality we have used $k! \geq (\frac{k}{e})^k$. Also note that 
\begin{equation}\label{eq:2}
    \left \lVert \sum_{k=0}^{m} \frac{(-i H(\g) t)^k}{k!} \right \rVert \leq \sum_{k=0}^{m} \frac{\lVert  H(\g)\rVert^k t^k}{k!} \leq e^{\lVert H(\g)\rVert t}
\end{equation}
Then at any $\g\in \mathbb{R}^l$
\begin{align}
    &|D(\g)-\mathcal{T}_{2m}(\g)|\\
    &= \left||\langle \psi_{out} |e^{-iH(\g)t} |\psi_{in} \rangle |^2 -  |\bra{\psi_{out}}\sum_{k=0}^{m} \frac{(-iH(\g)t)^k}{k!}\ket{\psi_{in}}|^2  \right|\\
    &= \left|\Tr\left(\bra{\psi_{out}}\otimes\bra{\psi_{in}}  e^{-iH(\g)t}\otimes  e^{iH(\g)t} - \sum_{k=0}^{m} \frac{(-iH(\g)t)^k}{k!} \otimes \left(\sum_{k=0}^{m} \frac{(-iH(\g)t)^k}{k!}\right)^{\dagger}\ket{\psi_{in}}\otimes \ket{\psi_{out}} \right ) \right|\\
    &\leq \left \lVert  e^{-iH(\g)t}\otimes  e^{iH(\g)t} - \sum_{k=0}^{m} \frac{(-iH(\g)t)^k}{k!} \otimes \left(\sum_{k=0}^{m} \frac{(-iH(\g)t)^k}{k!}\right)^{\dagger} \right\rVert\\
    &\leq \lVert e^{-iH(\g)t} \rVert \left\lVert e^{-iH(\g)t} - \sum_{k=0}^{m} \frac{(-iH(\g)t)^k}{k!}\right\rVert +  \left\lVert  \sum_{k=0}^{m} \frac{(-iH(\g)t)^k}{k!} \right\rVert  \left\lVert (e^{-iH(\g)t})^{\dagger} - (\sum_{k=0}^{m} \frac{(-iH(\g)t)^k}{k!})^{\dagger} \right\rVert\\
    & \leq 2 \max\left(\left\lVert  \sum_{k=0}^{m} \frac{(-iH(\g)t)^k}{k!} \right\rVert, \lVert e^{-iH(\g)t} \rVert\right) \left\lVert e^{-iH(\g)t} - \sum_{k=0}^{m} \frac{(-iH(\g)t)^k}{k!} \right\rVert\\
    &\leq 2e^{\lVert H(\g) \rVert t} \sum_{k=m+1}^{\infty} \left(\frac{e\lVert H(\g) \rVert t}{m}\right)^k\label{eq:3}
\end{align}
where all norms are the spectral norm. Here the first inequality is by H\"{o}lder's and the second inequality is by the following fact: for any matrices $A,B, C,D$ of the appropriate dimensions,
\begin{equation}
    \lVert A \otimes B - C\otimes D\rVert = \lVert A \otimes B - A\otimes D + A\otimes D - C\otimes D\rVert \leq \lVert A \otimes B - A\otimes D \rVert + \lVert A\otimes D - C\otimes D\rVert 
\end{equation}
for any subadditive matrix norm. The last inequality is by \Cref{eq:1} and \Cref{eq:2}. If $m \geq e\lVert H(\g)\rVert t$, then we may further bound \Cref{eq:3} using the standard formula for a geometric series, $\sum_{k=m+1}^{\infty} c^k = c^{m+1}/(1-c)$ for $c<1$, as
\begin{equation}
    2e^{\lVert H(\g) \rVert t}  \left(\frac{e\lVert H(\g) \rVert t}{m}\right)^{m+1}
\end{equation}

Let $c := \frac{e\lVert H(\g) \rVert t}{m}$, $c <1$. The above expression for the error of the Taylor approximation will be upper bounded by $\eps$ as long as the polynomial degree $m$ is chosen as
\begin{align}
    m+1 \geq \log_c(\eps e^{-\lVert H(\g) \rVert t}) = \frac{\log(1/\eps) +\lVert H(\g) \rVert t}{\log(1/c)}.
\end{align}
Combining the restrictions on $m$, and simplifying, yield the first part of the theorem \Cref{eq:polydegree}.
\end{proof}

Fix a polynomial degree $m$ and accuracy $\epsilon>0$. The polynomial $\mathcal{T}_{2 m}$ is an $\epsilon$-approximation to $D$ only on a restricted portion of $D$'s domain -- namely, the set of $\bf{g}$ such that \Cref{eq:polydegree} is satisfied. However, using concentration inequalities for Matrix Gaussian series \cite{Tropp15}, one could show that $\calT_{2m}$ is a good approximation to $D$ with overwhelming probability over $\g \sim \calN_l.$ We omit the formal statement and proof as it is not needed in this paper. 

\section{Complexity without hiding}

\subsection{Overview of argument}\label{sec:overview}
Our experimental task is \emph{sampling} from the output distribution of a random Hamiltonian evolution at constant time \(\tau=O(1)\). Our hardness proof, however, centers around the task of \emph{evaluating} specific output probabilities. This subsection explains, at a high level, why hardness of average-case \emph{evaluation} suffices to rule out efficient classical \emph{sampling}. In adapting this proof to random Hamiltonians, we address a technical obstacle: the lack of the ``hiding'' property that underpins analogous arguments for random circuit sampling.

\paragraph{Notation in this subsection.}
Let \(H\) be an \(n\)-qubit Hamiltonian, \(\ket{\psi_\beta}\) an efficiently preparable input state, and let measurements be performed in the \(X\) basis. For \(x\in\{0,1\}^n\), denote by \(\ket{x}_X := \bigotimes_{i=1}^n \ket{(-1)^{x_i}} \in \{\ket{+},\ket{-}\}^{\otimes n}\) the \(X\)-basis string. We write
\[
p_{H,y}(x)\ :=\ \left|\!\bra{x}_X\, e^{-iH\tau}\, Z^y\,\ket{\psi_\beta}\!\right|^2,\qquad
Z^y:=\bigotimes_{i:y_i=1} Z_i.
\]
Thus \(D_{H,y}\) is the distribution on outcomes \(x\) induced by the experiment.

\begin{definition}[Approximate (Hamiltonian) sampler]\label{def:sampler}
A (worst-case) \(\varepsilon\)-\emph{approximate sampler} at time \(\tau=O(1)\) is a probabilistic polynomial-time algorithm \(\mathcal{S}\) which, on input \((H,y,1^{1/\varepsilon})\), outputs samples from a distribution \(D'_{H,y}\) satisfying
\(
\|D'_{H,y}-D_{H,y}\|_1 \le \varepsilon
\)
for every \(n\)-qubit Hamiltonian \(H\) in the family under consideration and every \(y\in\{0,1\}^n\).
\end{definition}

\begin{definition}[Average-case approximate evaluator]\label{def:evaluator}
Fix a distribution \(\mathcal{D}\) over instances \(x\). A polynomial-time algorithm \(\mathcal{O}\) is an \emph{average-case} \((\varepsilon,\delta)\)-\emph{approximate evaluator} for a quantity \(p(x)\) with respect to $\mathcal{D}$ if
\[
\Pr_{x\sim \mathcal{D}}\!\left[\,\big|\mathcal{O}(1^{1/\varepsilon},1^{1/\delta},x)-p(x)\big| \le \tfrac{\varepsilon}{2^n}\,\right] \ge 1-\delta.
\]
When \(x\) parameterizes a Hamiltonian \(H(x)\), we call \(\mathcal{O}\) an average-case approximate \emph{Hamiltonian} evaluator if \(p(x)=\big|\!\bra{+^n}e^{-iH(x)\tau}\ket{+^n}\!\big|^2\).
\end{definition}

\paragraph{High-level pipeline.}
Our argument follows the now-standard three-step framework (cf.\ random circuit \cite{BFNV18,BFLL21}/boson \cite{AA11,BDFH24}/fermion sampling \cite{Oszmaniec_2022}), specialized and adapted to Hamiltonians:

\begin{enumerate}
\item \textbf{Worst-case \(\#\)P-hardness of evaluating probabilities.} We define a worst-case hard family \(\mathcal{H}_{\mathrm{\text{worst}}}\) of constant-time nearest-neighbor Ising evolutions (Procedure~\ref{proc:worstH}), and show that computing \(\left|\!\bra{+^n}e^{-iH_\mathcal{S}\tau}\ket{+^n}\!\right|^2\) to additive error \(2^{-n}\) is \(\#\)P-hard in the worst case (Theorem~\ref{thm:worst}).

\item \textbf{Worst-to-average reduction for probabilities.} Our main result is a worst-to-average reduction from \(\mathcal{H}_{\mathrm{\text{worst}}}\) to our Gaussian geo-local Hamiltonian ensemble \(\mathcal{E}(2)\): if one can, on average over \(H\sim\mathcal{E}(2)\), approximate \(\left|\!\bra{+^n}e^{-iH\tau}\ket{+^n}\!\right|^2\) within \(2^{-\Theta(n\log n)}\), then one can (with high probability) approximate the same quantity for any worst-case \(H_\mathcal{S}\in\mathcal{H}_{\mathrm{\text{worst}}}\) to additive error \(2^{-n}\) (Theorems~\ref{thm:outer_2} and~\ref{thm:outer_3}). This yields average-case \(\#\)P-hardness of evaluation up to additive precision $2^{-\Theta (n\log n)}$ over \(\mathcal{E}(2)\).

\item \textbf{From sampling to evaluation in \(\textsf{BPP}^{\textsf{NP}^{\calS}}\).} Given a sampler \(\mathcal{S}\) (Def.~\ref{def:sampler}), Stockmeyer’s approximate counting turns oracle access to \(\mathcal{S}\) into multiplicative estimates of its output probabilities. A Markov-inequality argument then yields \emph{additive} \(O(2^{-n})\)-accurate estimates of the \emph{true} probabilities with high probability over outcomes \(x\). This must then be converted into an evaluator that succeeds over most instances, and this evaluator is in $\textsf{BPP}^{\textsf{NP}^{\calS}}$ (Lemma~\ref{lem:samplerimpliesevaluator}), often via what is known as a hiding property.
\end{enumerate}

\paragraph{Overcoming the lack of hiding for typical Hamiltonians in Step 3.}
For random circuits, one converts “success over most \emph{outcomes} for a fixed circuit” into “success over most \emph{circuits} for a fixed outcome” via a \emph{hiding} symmetry: appending random \(X\)’s permutes outputs while preserving the distribution over circuits (problem instances). Random Hamiltonians do not enjoy this symmetry: consider that for $H\sim \mathcal{E}(2)$ and $y\in \{0,1\}^n$, the random matrix $e^{-iHt}Z^y$ does not have the same distribution as $e^{-iHt}$. Our remedy is to \emph{randomize the input} by drawing \(y\sim\{0,1\}^n\) and preparing \(Z^y\ket{\psi_\beta}\). We then exploit a simple but crucial invariance:

\begin{quote}
\emph{If \(H\sim\mathcal{E}(2)\), then for any \(y\in\{0,1\}^n\), the conjugated Hamiltonian \(Z^yHZ^y\) has the same distribution as \(H\).}
\end{quote}

The upshot is a `joint hiding' property (\Cref{thm:hiding}) which holds jointly over the input (parametrized by $y\sim \text{Unif}(\{0,1\}^n)$) and the distribution over Hamiltonian instances ($H\sim \calE(2)$). 
Using the identity
\(
Z^y e^{-iH\tau} = e^{-i(Z^yHZ^y)\tau} Z^y,
\)
we successfully map ``success over most outcomes” to ``success over most random input-instance pairs” (Theorem~\ref{thm:prfixedoutput}). This replaces circuit hiding and completes the sampling\(\Rightarrow\)evaluation step for random Hamiltonians.

\medskip
Returning to the three-step pipeline, steps 1 and 2 yield average-case \(\#\)P-hardness of evaluation over \(\mathcal{E}(2)\) up to additive precision $2^{-\Theta (n\log n)}$, and this is proven more rigorously in \Cref{sec:sharpP}. This provides strong evidence in support of \Cref{conj:avcasesP}, which is the same statement but with a more generous additive precision of $2^{-\Theta (n)}$. Step 3 says that the above task (with additive precision $2^{-\Theta (n)}$) is in $\textsf{BPP}^{\textsf{NP}^{\calS}}$ -- this is proven more rigorously in \Cref{sec:BPPNP}. Putting these two statements together yields the contradiction (up to \Cref{conj:avcasesP})\footnote{While there is still a gap between the additive precision up to which we can prove $\#P$ hardness and that required by \Cref{conj:avcasesP}, we note that all existing proposals for quantum advantage via sampling have a similar gap. For some more established proposals this gap has been narrowed, for example in \cite{BDFH24}.}:

\begin{tcolorbox}[
  colback=blue!10,
  colframe=blue!80!black,
  rounded corners,
  shadow={drop},
  fontupper=\bfseries,
  boxrule=3pt,
  arc=5mm
]
Average-case approximate Hamiltonian evaluation is \(\#\)P-hard, yet an efficient sampler $\calS$ -- worst-case or average-case over the same ensemble -- would place it in \(\textsf{BPP}^{\textsf{NP}^\calS}\).
\end{tcolorbox}
By Toda’s theorem, this would collapse the polynomial hierarchy. Assuming PH does not collapse, we conclude that no such sampler exists.

\subsection{Average-case evaluation of output probabilities is in $\# \textsf{P}$}\label{sec:sharpP}

In this section, we argue that average-case evaluation of output probabilities of Hamiltonian time evolutions is in $\# \textsf{P}$. This will cover the first two steps of the three-step pipeline outlined in the `Overview' subsection, where step 2 contains the technical meat of our contribution and is the subject of \Cref{sec:worst-av}.

Ref.\cite{BermejoVega18} proposed a family of Ising-quench experiments, whose output probabilities are worst-case hard to evaluate up to exponentially small additive precision.
\begin{enumerate}
    \item Arrange $n = m\times m'$ qubits side-by-side on an $m$ row, $n$ column square lattice, with vertices $V$ and edges $E$, and prepare the product-state input 
\begin{equation}
\left|\psi_{\bm{\beta}}\right\rangle=\bigotimes_{i=1}^n\left(|0\rangle+\mathrm{e}^{\mathrm{i} \beta_{i}}|1\rangle\right), 
\end{equation}
parametrized by vector $\bm{\beta}$ where each $\beta_i \in \{0,\pi/4\}.$     
\item Let the system time-evolve under a nearest-neighbor translationally-invariant Ising Hamiltonian 
    \begin{equation}
H:=\sum_{(i, j) \in E} J_{ij} Z_i Z_j-\sum_{i \in V} h_i Z_i .
    \end{equation}
    and a fixed evolution time $\tau = 1$. Here, the parameters $J_{ij}, h_i = O(1)$.
    \item Measure all qubits in the $X$ basis.
\end{enumerate}
They map this architecture to a family of random circuits which, in turn, can implement dense Instantaneous Quantum Polynomial-time (IQP) circuits on a register via a measurement-based quantum computation (MBQC) argument. A consequence is that classical computation of these output probabilities is $\# \mathsf{P}$-hard, up to exponentially-small additive error. Intuitively, the hardness arises from the fact that these constant–depth evolutions can encode amplitudes of arbitrary polynomial-time quantum computations, and so even very fine additive estimates of those amplitudes would solve $\# \textsf{P}$–hard problems.

We adapt the worst-case hardness results of Ref.\cite{BermejoVega18} by pushing the hardness into the choice of Hamiltonians, rather than the input state. The main observation is that, in the prescription of Ref.\cite{BermejoVega18}, for indices $i\in [n]$ such that $\beta_i=1$, the corresponding single-qubit input state can be prepared via a $Z$-rotation acting on a $\ket{+}$ state:
\begin{equation}
    \ket{\psi_{\beta_i}} =\frac{\ket{0}+e^{i\pi/4}\ket{1}}{\sqrt{2}} = R_Z(\pi/4)\ket{+}.
\end{equation}
To prepare the full $n$-qubit input state $\ket{\psi_{\bm{\beta}}} =\bigotimes_{i=1}^n \ket{\psi_{\beta_i}} $ , then, one would have to initialize all states in $\ket{+}$, then act with a tensor product of $R_Z$ rotations on the qubits in $\calS$. But this may be re-written as a 1-local Hamiltonian time-evolution:
\begin{equation}\label{eq:1local}
    \bigotimes_{j\in \calS} R_Z^{(j)}(\theta) = \exp\left(-i\frac{\theta}{2}\sum_{j\in \calS} Z_j\right).
\end{equation}
Since Pauli $Z$s commute with Ising Hamiltonians, we may consider the composition of the $R_Z$ rotations used in state preparation, with the subsequent Ising Hamiltonian evolution, as the time evolution of a single effective Hamiltonian: the sum of the original Ising Hamiltonian with the 1-local Z-Hamiltonian of \Cref{eq:1local}. The full procedure is: 

\begin{procedure}[$\mathcal{H}_{\text{worst}}$, a worst-case hard class of Hamiltonian time-evolutions]\label{proc:worstH}
Consider the following procedure:
\begin{enumerate}
    \item Arrange $n = m\times m'$ qubits side-by-side on an $m$ row, $n$ column square lattice, with vertices $V$ and edges $E$, and prepare the product-state input $\ket{+}^{\otimes n}.$
    \item Pick $\calS \subseteq n$; let the system time-evolve for fixed evolution time $\tau = 1$, under a nearest-neighbor translationally-invariant Ising Hamiltonian 
    \begin{equation}
H_{\calS}:=\sum_{(i, j) \in E} J_{ij} Z_i Z_j-\sum_{i \in V} h_i Z_i + \frac{\pi}{8\tau}\sum_{j\in \calS} Z_j,
    \end{equation}
    with fixed parameters $J_{ij}, h_i = O(1).$ 
    \item Measure all qubits in the $X$ basis.
\end{enumerate}
\end{procedure}


Since every choice of $H_{\calS}$ describes exactly the same time evolution as that resulting from some choice of input state in \cite{BermejoVega18}, the following theorem follows immediately from their results:
\begin{theorem}[Worst-case hardness of Hamiltonian ensemble $\calH_{\text{worst}}$]\label{thm:worst}
    Fix evolution time $\tau=O(1)$. Let $\mathcal{H}_{\text{worst}} = \{H_{\calS}\}_{\calS \subseteq [n]}.$ Upon input of any Hamiltonian $H_\calS \sim \mathcal{H}_{\text{worst}}$, it is $\# \textsf{P}$-hard, in the worst-case, to approximately estimate the probability that time evolution corresponding to $H_{\calS}$ via \Cref{proc:worstH} will output the bitstring $+^n$, i.e.
    \begin{equation}
        D_{H_\calS}:=|\bra{+^n} e^{-iH_{\calS}\tau} \ket{+^{n}}|^2,
    \end{equation} 
    to additive error $2^{-\Theta(n)}.$
\end{theorem}

Leveraging the results of \Cref{sec:worst-av}, we can convert the above {\em worst-case} hardness statement to an average-case hardness statement over the ensemble $\calE(2)$:

\begin{theorem}[Worst-to-average-case reduction for evaluating probabilities (\Cref{thm:outer}, adapted)]\label{thm:outer_2}
Assume access to an average-case solver $\calA$ that, upon input $H\sim \calE(2)$, outputs $\hat{D}_H$ such that 
\begin{equation}
    \Pr_{H\sim \calE(2)}[|\hat{D}_H- |\bra{+^n} e^{-iH\tau} \ket{+^{n}}|^2\leq 2^{-\Theta(n\log n)}] \geq 1-\delta.
\end{equation}
Then for any $H_{\calS}\in \calH_{\text{worst}}$, with probability at least $1-O(\delta^{1/4})$, we can in $\textsf{BPP}$ output $\hat{p}$ satisfying
    \begin{equation}
        |\hat{p}- |\bra{+^n} e^{-iH_{\calS}\tau} \ket{+^{n}}|^2| \leq 2^{-\Theta(n)}.
    \end{equation}
\end{theorem}

\Cref{thm:outer_2} says that access to an average-case solver (over $\calE(2)$) allows to estimate the probabilities of worst-case time-evolution over $\calH_{\text{worst}}$, a $\,\# \textsf{P}$-hard task according to \Cref{thm:outer_2}. We conclude that it must be hard to implement the assumed average-case solver $\calA$, or more rigorously:
\begin{theorem}[Average-case \(\#\)P-hardness over \(\mathcal{E}(2)\)]
    It is $\# \textsf{P}$-hard to, upon input $H\sim \calE(2)$, output $\hat{D}_H$ such that 
\begin{equation}
    \Pr_{H\sim \calE(2)}[|\hat{D}_H- |\bra{+^n} e^{-iH\tau} \ket{+^{n}}|^2|\leq 2^{-\Theta(n\log n)}] \geq 1-\delta.
\end{equation}
\end{theorem}

In fact, we will need an even stronger average-case hardness statement as part of our hiding workaround, where we consider choosing not only a random time evolution, but also a random input state from an ensemble of product states we now define:
\begin{equation}
    \calE_{\psi}:=\{\ket{\psi_y}:=Z^{\bf y}\ket{+^n}\}_{{\bf y}\in \{0,1\}^n},
\end{equation}
so that $\ket{\psi_{\bf 0}}=\ket{+^n}.$ For every $H_{\calS}\in \calH_{\text{worst}}$, given a particular input state parametrized by $y$, we can construct a `shifted' Hamiltonian $H_{\calS,y}$ such that
\begin{equation}\label{eq:equivalence}
    |\bra{+^n} e^{-i\tau H_{\calS,y}}\ket{\psi_y}|^2 = |\bra{+^n} e^{-i\tau H_{\calS}}\ket{\psi_0}|^2; 
\end{equation}
one can check that 
\begin{equation}
    H_{\calS,y} = H_{\calS} +\frac{\pi}{2\tau}\sum_{k=1}^{n}y_k Z_k
\end{equation}
works. So for every input state $\ket{\psi_y}$ we may associate a corresponding worst-case `shifted' ensemble,
\begin{equation}\label{eq:H_y}
    \calH_{\text{worst}, y} := \{H_{\calS,y}:\,H_{\calS}\in \calH_{\text{worst}}\}.
\end{equation}
For each $y$, the corresponding shifted ensemble has the following property:
\begin{corollary}[Worst-case hardness of shifted ensembles \label{corr:worst_shifted}]
Fix any \(y\in\{0,1\}^n\). For \(\tau=O(1)\), approximating
\(
D_{H_{\mathcal{S}},y}:=\left|\!\bra{+^n}e^{-iH_{\mathcal{S},y}\tau}\ket{\psi_y}\!\right|^2
\)
to additive error \(2^{-\Theta(n)}\) is \(\#\)P-hard in the worst case over \(H_{\mathcal{S},y}\in\calH_{\text{worst}, y}\).
\end{corollary}
\begin{proof}
    This follows from \Cref{eq:equivalence} and \Cref{thm:worst}.
\end{proof}
Now we notice that \Cref{thm:outer_2} is agnostic to the choice of input state in $\calE_{\psi}.$ More precisely, 
\begin{theorem}[Worst-to-average-case reduction with arbitrary inputs]\label{thm:outer_3}
Fix $y\in \{0,1\}^n$ (which fixes the input state). Assume access to an average-case solver $\calA$ that, upon input $H\sim \calE(2)$, outputs $\hat{D}_H$ such that 
\begin{equation}
    \Pr_{H\sim \calE(2)}[|\hat{D}_H- |\bra{+^n} e^{-iH\tau} \ket{\psi_y}|^2\leq 2^{-\Theta(n\log n)}] \geq 1-\delta.
\end{equation}
Then for any $H_{\calS,y}\in \calH_{\text{worst}, y}$, with probability at least $1-O(\delta^{1/4})$, we can in poly$(n)$ time output $\hat{p}$ satisfying
    \begin{equation}
        |\hat{p}- D_{H_{\calS,y}}| \leq 2^{-\Theta(n)}.
    \end{equation}
\end{theorem}

Combining Corollary~\ref{corr:worst_shifted} and Theorem~\ref{thm:outer_3} yields average-case \(\#\)P-hardness when both the Hamiltonian and the input are randomized:

\begin{theorem}[Average-case (over $H\sim\calE(2), \, \ket{\psi_y}\sim\calE_{\psi}$) evaluation of probabilities is $\# \mathsf{P}$-hard]\label{thm:avcasesP}
    It is $\# \textsf{P}$-hard to, upon input $H\sim \calE(2)$ and $y\sim \text{Unif}\{0,1\}^n$, output $\hat{D}_H$ such that 
\begin{equation}
    \Pr_{H\sim \calE(2),y\sim \text{Unif}\{0,1\}^n}[|\hat{D}_H- |\bra{+^n} e^{-iH\tau} \ket{\psi_y}|^2\leq 2^{-\Theta(n\log n)}] \geq 1-\delta.
\end{equation}
\end{theorem}
\begin{proof}
    Suppose to the contrary there were an algorithm $\mathcal{A}$ that evaluates $|\bra{+^n} e^{-iH\tau} \ket{\psi_y}|^2$ with probability at least $1-\delta$ over $H\sim \mathcal{E}(2)$, $y \sim \text{Unif}(\{0,1\}^n)$. There must exist at least one $y'\in \{0,1\}^n$ for which $\mathcal{A}$ succeeds in evaluating $|\bra{+^n} e^{-iH\tau} \ket{\psi_{y'}}|^2$ with probability at least $1-\delta$ over $H\sim \mathcal{E}(2)$ while the input state is held fixed at $\ket{\psi_{y'}}$. But by \Cref{thm:outer_3}, for any $H_{\calS, y'} \in \calH_{\text{worst},y'}$, with probability at least $1-O(\delta^{1/4})$, we can in $\poly(n)$ time approximate $D_{H_{\calS,y'}}$ to precision $2^{-\Theta(n)}$. But this contradicts \Cref{corr:worst_shifted}.
\end{proof}

\Cref{thm:avcasesP} gives strong evidence in support of our main conjecture which concerns evaluation of probabilities up to a smaller additive precision of $2^{\Theta(n)}$. This value of additive precision is necessary in order to connect up with \Cref{thm:samplerimpliesevaluator}.

\begin{conjecture}[Average-case (over $H\sim\calE(2), \, \ket{\psi_y}\sim\calE_{\psi}$) evaluation of probabilities up to precision $2^{-\Theta(n)}$ is $\# \mathsf{P}$-hard]\label{conj:avcasesP}
    It is $\# \textsf{P}$-hard to, upon input $H\sim \calE(2)$ and $y\sim \text{Unif}\{0,1\}^n$, output $\hat{D}_H$ such that 
\begin{equation}
    \Pr_{H\sim \calE(2),y\sim \text{Unif}\{0,1\}^n}[|\hat{D}_H- |\bra{+^n} e^{-iH\tau} \ket{\psi_y}|^2\leq 2^{-\Theta(n)}] \geq 1-\delta.
\end{equation}
\end{conjecture}

\subsection{Sampling implies average-case evaluation of output probabilities in $\textsf{BPP}^{\textsf{NP}^{\calS}}$}\label{sec:BPPNP}

This subsection proves that an algorithm that samples from the output of Hamiltonian time evolutions can be converted into an average-case evaluator in $\textsf{BPP}^{\textsf{NP}^{\calS}}$. This is the third step of the three-step pipeline outlined in the `Overview' subsection.

\begin{theorem}[Sampling \(\Rightarrow\) average-case evaluation\label{thm:samplerimpliesevaluator}]
The existence of a Hamiltonian sampler $\calS$ implies that average-case (over $H\sim \mathcal{E}(2)$, $\ket{\psi_y}\sim \mathcal{E}_{\psi}$) approximate evaluation of the probability of outputting $\ket{+^n}$ is in $\textsf{BPP}^{\textsf{NP}^{\calS}}$, i.e. there is an algorithm $\mathcal{A}$ in $\textsf{BPP}^{\textsf{NP}^{\calS}}$ such that
\begin{equation}
    \Pr_{H\sim \calE(2),y\sim \text{Unif}\{0,1\}^n}[|\mathcal{A}(H,y)-|\bra{+^n} e^{-iH\tau} \ket{\psi_y}|^2|\leq \eps/2^n] \geq 1-\gamma.
\end{equation}
\end{theorem}
While Appendix A.4 of \cite{BFNV18} has a similar statement to our \Cref{thm:samplerimpliesevaluator} (Theorem 22), they took an average only over circuits, whereas our average-case is over both the input and the evolution. The difference arises random circuits enjoy a `hiding' property which random Hamiltonian evolutions do not. We elaborate on this in due course.

As an intermediate step, we observe that the existence of a worst-case classical sampler already implies that for any given Hamiltonian, the probability of a random outcome $y \sim \{0,1\}^n$ can be estimated with additive precision $\eps/2^n$. This immediately follows from similar arguments in the literature, for example Lemma 23 of \cite{BFNV18}. The proof uses Stockmeyer's reduction \cite{Stockmeyer85} to obtain a relative error approximation, and then uses Markov's inequality to translate that into the desired additive error approximation. 
\begin{lemma}[From worst-case sampling to average-case (over outcomes) evaluation\label{lem:samplerimpliesevaluator}]
    Given a worst-case (over 2-geolocal Hamiltonians) $\eps$-approximate Hamiltonian sampler $\mathcal{S}$ at fixed $\tau=O(1)$, there exists an average-case (over outcomes) approximate Hamiltonian evaluator in $\textsf{BPP}^{\textsf{NP}^{\calS}}$, that is, one that, upon input of the desired bitstring $x \in \{0,1\}^n$, outputs $\tilde{q}_{x}$ such that
    \begin{equation}\label{eq:lemma23}
        \Pr_{x\sim \{0,1\}^n}[|\tilde{q}_{x} - |\bra{+^{x}}e^{-iHt}\ket{\psi_{\beta}}|^2| \geq \eps/2^n]\leq \delta.
    \end{equation}
\end{lemma}

To continue from this step, the proofs of average-case hardness for random circuit sampling crucially leverage what they call the ``hiding property" of random circuits, which states that the distribution over random circuits is invariant under appending any $\bigotimes_{i:y_i=1}X^{y_i}$ after the last layer. Noting that $X$ flips a computational basis state, this property allows them to convert success over a random outcome (given by \Cref{lem:samplerimpliesevaluator}) into success over a random circuit. 

Unfortunately, hiding fails for random Hamiltonians: for $y\sim \{0,1\}^n$, letting 
\[
Z^y := \bigotimes_{i:y_i=1} Z_i,
\]
for $H\sim \mathcal{E}(2)$, the random matrix $e^{-iHt}Z^y$ at $t=O(1)$ does not have the same distribution as $e^{-iHt}$. To remedy this, we will instead make the following observation: 

\begin{theorem}[Random 2-local Hamiltonians together with random input state satisfy hiding\label{thm:hiding}]
Given 
any $x\in \{0,1\}^n$, let 
\[
Z^x := \bigotimes_{i:x_i=1} Z_i.
\]
We have
\begin{equation}
    \Pr_{(x,H)\sim\{0,1\}^n\times \mathcal{E}(2)}[(x,H)] = \Pr_{(x,H)\sim\{0,1\}^n\times \mathcal{E}(2)}[(x,Z^xHZ^x)]
\end{equation}
\end{theorem}

\begin{proof}
Given any 2-local Pauli term $P_iP_j$, and for any $Z^{x_i}Z^{x_j}$, conjugating the term by $Z^{x_i}Z^{x_j}$ either does not change the term or simply flips its sign:
\begin{equation}
Z^{x_i}Z^{x_j}(P_iP_j)Z^{x_i}Z^{x_j} = \pm P_iP_j.
\end{equation}
Since $H$ is a sum of 2-local Pauli terms, each with an independent Gaussian-distributed coefficient, conjugation by any $Z^x$ flips the sign of some of the coefficients of $H$. However, because the multivariate Gaussian is symmetric about zero, the distribution of coefficients of $Z^x H Z^x$ is identical to that of $H$. Formally, 
\begin{align}
    \Pr_{(x,H)\sim\{0,1\}^n\times \mathcal{E}(2)}[(x,Z^xHZ^x)] &= \Pr_{x\sim \{0,1\}^n}[x \Pr_{H\sim \mathcal{E}(2)}[Z^x H Z^x]]= \Pr_{x\sim \{0,1\}^n}[x \Pr_{H\sim \mathcal{E}(2)}[H]]\\
    &= \Pr_{x\sim \{0,1\}^n}[x]\Pr_{H\sim \mathcal{E}(2)}[H]
\end{align}
as stated.
\end{proof}

This finally allows us to prove \Cref{thm:prfixedoutput}, which connects up with the $\# \textsf{P}$-hardness statement of the previous subsection. 
\begin{theorem}[Converting success over random outcomes to success over random inputs and Hamiltonians\label{thm:prfixedoutput}]
    Any approximate Hamiltonian evaluator $\mathcal{A}$ that succeeds for most outcomes will also succeed over a random input state and Hamiltonian, for a fixed outcome. 
    
    i.e. any approximate Hamiltonian evaluator 
    fulfilling the guarantees of \Cref{eq:lemma23} will also output an $\epsilon/2^n$-estimate of $|\langle +^{n}|e^{-iHt}|\psi_{\beta,y}\rangle|^2$ with probability $1-\delta$ over $H\sim \mathcal{E}(2)$, $y\sim \{0,1\}^n.$
\end{theorem}
\begin{proof}
    For any evaluator $\mathcal{A}$, call a probability 
    $D_{x,H,y}:=\left|\bra{+^n}Z^x e^{-iHt}Z^y \ket{\psi_{\beta}}\right|^2$ \emph{good} if the estimate of that probability output by $\mathcal{A}$, $
    \hat{D}_{x,H,y}$ satisfies 
\begin{equation}
\left|\hat{D}_{x,H,y} - D_{x,H,y} \right| < \epsilon/2^n.
\end{equation}
\Cref{lem:samplerimpliesevaluator} states that if an approximate sampler $\mathcal{S}$ exists, then the associated evaluator $\mathcal{A}$ succeeds over most outcomes,
\begin{align}\label{eq:ev1}
\forall H\in \mathcal{E}(2),\quad  \Pr_{x\sim \{0,1\}^n}\Big[D_{x,H,\mathbf{0}} \,\, \text{is good} \Big] \geq 1-\delta.
\end{align}
Note, however, that 
\begin{align}
D_{x,H,\mathbf{0}} &:= \left|\bra{+^n} Z^x e^{-iHt}\ket{\psi_{\beta}}\right|^2 = \left|\bra{+^n} e^{-iZ^x H Z^x t} Z_x\ket{\psi_{\beta}}\right|^2 \\
&= D_{\mathbf{0},Z^xHZ^x,x}.
\end{align}
where we have used the identity:
\begin{equation}
Pe^{-iHt} = e^{-iPHPt}P
\end{equation}
for any $n$-qubit Pauli $P$ and any Hamiltonian $H$ consisting of Pauli terms. Thus \Cref{eq:ev1} implies: 
\begin{align}
&\forall H\in \mathcal{E}(2),\quad \Pr_{x\sim \{0,1\}^n}\Big[D_{x,H,\mathbf{0}} \text{ is good }\Big] \geq 1-\delta\\
&\rightarrow \forall H\in \mathcal{E}(2),\quad \Pr_{x\sim \{0,1\}^n}\Big[D_{\mathbf{0},Z^xHZ^x,x} \text{ is good }\Big] \geq 1-\delta\\
&\rightarrow  \Pr_{x\sim \{0,1\}^n, \,H\sim \mathcal{E}(2)}\Big[D_{\mathbf{0},Z^xHZ^x,x} \text{ is good}\Big] \geq 1-\delta\\
&\rightarrow \Pr_{x\sim \{0,1\}^n, \,H\sim \mathcal{E}(2)}\Big[D_{\mathbf{0},H,x} \text{ is good}\Big] \geq 1-\delta
\end{align}
where the last implication is by \Cref{thm:hiding}.
\end{proof}

\section{Algorithmic robust Berlekamp-Welch}
\subsection{Polynomial interpolation lemmas}
For the polynomial interpolation, we cannot use a direct adaptation of Berlekamp-Welch, as that requires the correct points to be exact evaluations of the polynomial whereas we only assume our average-case solver can approximately evaluate the polynomial. Instead we will use

\begin{lemma}[Discrete Remez inequality.]\label{lem:Remez} Let $\left\{x_j\right\}_{j=0}^d \subset[0,b]$ be a $\delta$-separated set of points, meaning that $\left|x_i-x_j\right| \geq \delta$ for $i \neq j$. Then if $p$ is a degree-d polynomial and $L \geq b$,

$$
|p(L)| \leq\left(e^2(\delta d)^{-1} L\right)^d \max _{0 \leq j \leq d}\left|p\left(x_j\right)\right| .
$$
\end{lemma}
\begin{proof}
    Follows from, e.g., Lemma B.1 of \cite{BDFH24} via a variable rescaling $x'=x/b$ and $p'(x'b) = p(x)$ (where $p,x$ are from the original lemma and $p', x'$ are the new polynomial and variable of interest).
\end{proof}

We will also need a version of the Discrete Remez inequality where we interpolate to a point that is within the support of the set of evaluations we are given.

\begin{lemma}[{Modified Discrete Remez on $[-1,1]$}]
Let $\left\{x_j\right\}_{j=0}^d \subset[0,b]$ be a $\delta$-separated set of points.
If $p$ is a real polynomial of degree $\le d$ and $L\in[0,1]$, then
\[
\quad
|p(L)|\;\le\;\frac{2^{\,d}}{\delta^{\,d}\,d!}\;\max_{0\le j\le d}|p(x_j)|
\]
\end{lemma}

\begin{proof}
By the Lagrange interpolation formula,
\[
p(L)=\sum_{j=0}^d p(x_j)\prod_{k\ne j}\frac{L-x_k}{x_j-x_k}.
\]
Since $L,x_k\in[0,1]$, we have $|L-x_k|\le 1$, hence
\[
|p(L)|\le \max_j|p(x_j)|\sum_{j=0}^d \prod_{k\ne j}\frac{1}{|x_j-x_k|}.
\]
Reorder the nodes so that $x_0\le x_1\le \cdots\le x_d$. For a fixed $j$,
\[
\prod_{k\ne j}|x_j-x_k|
=\Big(\prod_{m=1}^{j} (x_j-x_{j-m})\Big)\Big(\prod_{n=1}^{d-j} (x_{j+n}-x_j)\Big)
\;\ge\; \Big(\prod_{m=1}^{j} m\delta\Big)\Big(\prod_{n=1}^{d-j} n\delta\Big)
=\delta^{\,d}\,j!\,(d-j)! .
\]
Therefore
\[
|p(L)|\le \max_j|p(x_j)|\,\delta^{-d}\sum_{j=0}^d \frac{1}{j!\,(d-j)!}
=\max_j|p(x_j)|\,\delta^{-d}\,\frac{1}{d!}\sum_{j=0}^d \binom{d}{j}
=\frac{2^{\,d}}{\delta^{\,d}\,d!}\,\max_j|p(x_j)|,
\]
as claimed.
\end{proof}

\begin{lemma}\label{lemma:leading-coeff}
Let $a,b\in \mathbb{R}$, and let $\left\{x_j\right\}_{j=0}^d \subset[a,b]$ be a $\delta$-separated set of points, meaning that $\left|x_i-x_j\right| \geq \delta$ for $i \neq j$. Then if $p$ is a degree-d polynomial with leading coefficient $1$, then there must be some point $x_j$ with
\[
|p(x_j)| \geq \frac{\delta^d}{d+1} \,.
\]
\end{lemma}
\begin{proof}
This follows by using Lagrange interpolation as above and then using triangle inequality to bound the leading coefficient.    
\end{proof}

\subsection{Robust Berlekamp-Welch}
We also have the following algorithmic variant of the Berlekamp–Welch algorithm that can tolerate additive errors in the values we are trying to interpolate. 

While standard Berlekamp-Welch solves a set of linear equations to find an ``error-locator" polynomial that vanishes exactly at corrupted positions, our robust algorithmic version of Berlekamp-Welch solves two linear programs, both of which have $O(n)$ variables and constraints and are solvable in poly $(n)$ time with standard interior-point methods. Because our non-corrupted evaluations can now have additive errors, the first linear program is to find a {\em soft} error-locator $s$ which only approximately vanishes at corrupted positions, and the second linear program fits a signal to $s$. 

\begin{theorem}[Robust efficient Berlekamp-Welch\label{thm:REBW}]
Let $x_1, \dots, x_n \in [-1,1]$ and let $y_1, \dots , y_n \in \R$.  Assume that $x_1,\dots , x_n$ are $\delta$-separated.  Let $0 <  k < n$ and assume that there exists a polynomial $p$ of degree at most $n - 2k - 1$ such that $|p(x_i) - y_i| \leq \eps$ for at least $n - k$ distinct indices $i$.  Then there is an algorithm that, on input $\{(x_i,y_i)\}_{i\in [n]}$, runs in time $\text{poly}(n)$ and outputs a polynomial $q$ of degree at most $n - 2k - 1$ such that  $|q(x_i) - p(x_i)| \leq (10/\delta)^{2n} \cdot \eps$ for at least $n-2k$ values $i \in [n]$.
\end{theorem}
\begin{proof}
Let $r(x) = r_{n-k-1} x^{n - k-1} + \dots + r_0$ be a polynomial of degree $n-k-1$ and $s(x) = x^{k} + s_{k-1}x^{k-1} + \dots + s_0$ be a polynomial of degree $k$.  We treat the coefficients $r_j, s_j$ as unknowns and then we set up the following linear program.  We enforce that
\[
\begin{split}
&|s_{k-1}| \leq \binom{k}{k-1} , \dots , |s_0| \leq \binom{k}{0} \\
&- 2^{k} \eps  \leq r(x_i) - y_i \cdot s(x_i) \leq 2^{k}  \eps \quad \forall i \in [n] 
\end{split}
\]
First, to see why this is feasible, let $\{ i_1, \dots , i_k \}$ contain the $k$ indices where $|p(x_i) - y_i| \geq \eps$.  Then let $s(x) = (x-x_{i_1}) \cdots (x - x_{i_k})$.  Let $r(x) = s(x) p(x)$. The first condition then follows from expanding $s(x)$ and using triangle inequality.  The second condition  holds because for $i \in \{i_1, \dots , i_k \}$, $r(x_i) - y_i \cdot s(x_i) = 0$ and otherwise 
\[
|r(x_i) - y_i \cdot s(x_i)| = |s(x_i)| \cdot |p(x_i) - y_i| \leq 2^k \eps \,.
\]

Now we consider any feasible solution to the linear program, say $r(x), s(x)$.  Given this feasible solution, we solve a second linear program for a polynomial $q(x) = q_{n - 2k - 1}x^{n - 2k - 1} + \dots + q_0$ such that 
\[
- (10/\delta)^n \eps \leq (q(x_i)  - y_i) s(x_i) \leq  (10/\delta)^n \eps \quad \forall i \in [n] \,.
\]
Note that now the coefficients of $s(x)$ are constants, fixed to the values of the solution from the first linear program. To see why this second program is feasible, consider  $q(x) = p(x)$.  Then by the feasibility of $r(x), s(x)$ in the first linear program, we know for all $i \in [n] \backslash \{i_1, \dots , i_k \}$,
\begin{align}
|r(x_i) - p(x_i)s(x_i)| &\leq |r(x_i) - y_i\cdot s(x_i)| + |s(x_i)| \cdot |p(x_i) - y_i| \\
&\leq (b-a)^k \eps + |s(x_i)| \cdot |p(x_i) - y_i| \leq 2^{k+1} \eps \,.
\end{align}
Now the polynomial $r(x) - p(x)s(x)$ has degree at most $n - k - 1$ and thus  Lemma~\ref{lem:ModifiedRemez} implies that for all $x \in [-1,1]$,
\[
|r(x) - p(x)s(x)| \leq (8/\delta)^n \eps \,.
\]
Combining with the feasibility of $r(x), s(x)$ in the first linear program, we deduce that for all $i \in [n]$,
\[
|(p(x_i) - y_i)s(x_i)| \leq (10/\delta)^n\eps \,.
\]
This completes the proof of feasibility for the second program. Now, we will show that any feasible polynomial $q$ in the second linear program will have the desired final property.  For such a feasible $q$, for all $i \in [n] \backslash \{i_1, \dots , i_k \}$,
\[
|(q(x_i) - p(x_i))s(x_i)| \leq (10/\delta)^n \eps + |s(x_i)| \cdot |p(x_i) - y_i| \leq (2^k + (10/\delta)^n)\eps \,.
\]
Finally, among all $i \in [n] \backslash \{i_1, \dots , i_k \}$, Lemma~\ref{lemma:leading-coeff} implies that there are at most $k$ indices such that $|s(x_i)| \leq \delta^k/(k+1)$.  Thus, there must be at least $n - 2k$ indices $i \in [n]$ such that 
\[
|q(x_i) - p(x_i)| \leq (10/\delta)^{2n} \eps \,.
\]
This completes the proof.
\end{proof}

We observe that our above theorem also immediately applies to other quantum advantage settings that rely on the hardness of approximately evaluating probabilities to argue for hardness of sampling, and consequently also improves the evidence for hardness there:

\begin{corollary}[$\# \textsf{P}$-hardness of SUPER in the saturated regime.] Let $\mid$ SUPER $\left.\right|_{ \pm} ^2$ be the problem of \textit{Sub-Unitary Permanent Estimation
with Repetitions}, defined in \cite{bouland2025complexitytheoreticfoundationsbosonsamplinglinear}. Let $m \geq 2.1 n$. In the regime $m=\Theta(n)|\mathrm{SUPER}|_{ \pm}^2$ is $\# \textsf{P}$-hard under $\textsf{BPP}$ reductions to additive error $\epsilon(S)=e^{-5 n \log (n)-O(n)}$. with probability at least $1-\delta$, with $\delta=1 / \operatorname{poly}((n)).$
\end{corollary}

\begin{corollary}[$\#\textsf{P}$ hardness of approximating random circuit output probabilities]
    Let $\mathcal{A}$ be a circuit architecture so that computing $\mathrm{p}_0(C) = |\bra{0}C\ket{0}|^2$ to within additive error $2^{-O(m)}$ is $\# \textsf{P}$-hard in the worst case. Then the following problem is $\# \textsf{P}$-hard under $\textsf{BPP}$ reductions: for any constant $\eta<\frac{1}{4}$, on input a random circuit $C \sim \mathcal{H}_{\mathcal{A}}$ with $m$ gates, compute the output probability $\mathrm{p}_0(C)$ up to additive error $\delta=\exp (-O(m \log m))$, with probability at least $1-\eta$ over the choice of $C$.
\end{corollary}
\section{Worst-to-Average-Case Reduction}\label{sec:worst-av}

\textbf{Notation:} In this subsection, we will use the symbol $\x$ in place of $\g$ which was previously used to denote the vector of coefficients of the Hamiltonian.

Recall that the function of interest is $D:\mathbb{R}^l\rightarrow [0,1]$ which is the probability that the time evolution of the Hamiltonian $H(\x)$ (see \Cref{eq:Hg}), initialized on the state $\ket{+^n}$, will output $\ket{+^n}$ when measured in the X basis:
\begin{equation}
    D(\x) = |\bra{+^n} e^{-iH(\bm \x)t} \ket{+^{n}}|^2,
\end{equation}
where $H(\x)\sim \calE(2).$ 

Given some worst-case instance from $\mathcal{E}_{worst}$, the goal of the worst-to-average-case reduction is to estimate the value of $D(\x_{worst})$, given access to an efficient algorithm $\mathcal{A}$ such that
\begin{equation}\label{eq:guarantee_A}
    \operatorname{Pr}_{\x \sim \mathcal{N}_l}\left[\left|\mathcal{A}(\x)-D(\x)\right| \leq \epsilon_{\mathcal{A}}\right] \geq 1-\delta.
\end{equation}
That is, $\mathcal{A}$ is an algorithm that succeeds ``on average". 

To summarize the situation, we have at our disposal the ability to run $\calA$, which succeeds with high probability on inputs sampled from $\calN_l$. The goal of the worst-to-average-case reduction is to show that we may run $\calA$ on some set of input points $\bf{x}_i\in \R^l$, and process the $\{\calA(\x_i)\}$ to obtain an estimate of $D(\g_{worst})$.This then implies that estimating $D(\g_{worst})$ is no harder than estimating $D(\x)$ for `most' $\x \sim \calN_l.$ In other words, even the worst-possible point is no harder than `most' of them. 

It is natural to use polynomial interpolation for the processing step, where given the knowledge that $\calT_m$ approximates $D$, we interpolate a degree $m$ polynomial through the known points $\{(\x_i,\calA(\x_i))\}$. We thus face a problem of robust multivariate polynomial interpolation of the polynomial $\calT_m: \R^l \rightarrow [0,1]$. But univariate polynomial interpolation with error correction is much more tractable. Could we try to compose the desired multivariate polynomial interpolation out of interpolations on lower-dimensional manifolds within $\mathbb{R}^l$? This would reduce the number of variables we have to interpolate in. However, this idea runs into a different problem: inputs restricted to lie in the same plane or circumference are by this restriction correlated with each other. So, adopting this approach means we cannot use $\mathcal{A}$'s success guarantee out-of-the-box, as the inputs restricted to the same manifold are no longer independent. 

Nevertheless, a variation of this idea eventually works. The key realization is that, if we feed $\calA$ points that are sampled from the marginal of $\calN_l$ on the desired manifold, a counting argument ensures that a randomly-sampled manifold will already be one on which $\calA$ succeeds with high probability, regardless of any correlations within that manifold. Using this, we show how to reduce multivariate polynomial interpolation to multiple rounds of univariate polynomial interpolation. It turns out that this seemingly formidable task in high-dimensional geometry can be accomplished by working on a single slice of $\mathbb{R}_l$, which  is visualized in \Cref{fig:slicingdicing}. For these reasons, we call our technique ``slicing and dicing the sphere".
\begin{figure}[!htbp]
    \centering
    \includegraphics[width=0.5\linewidth]{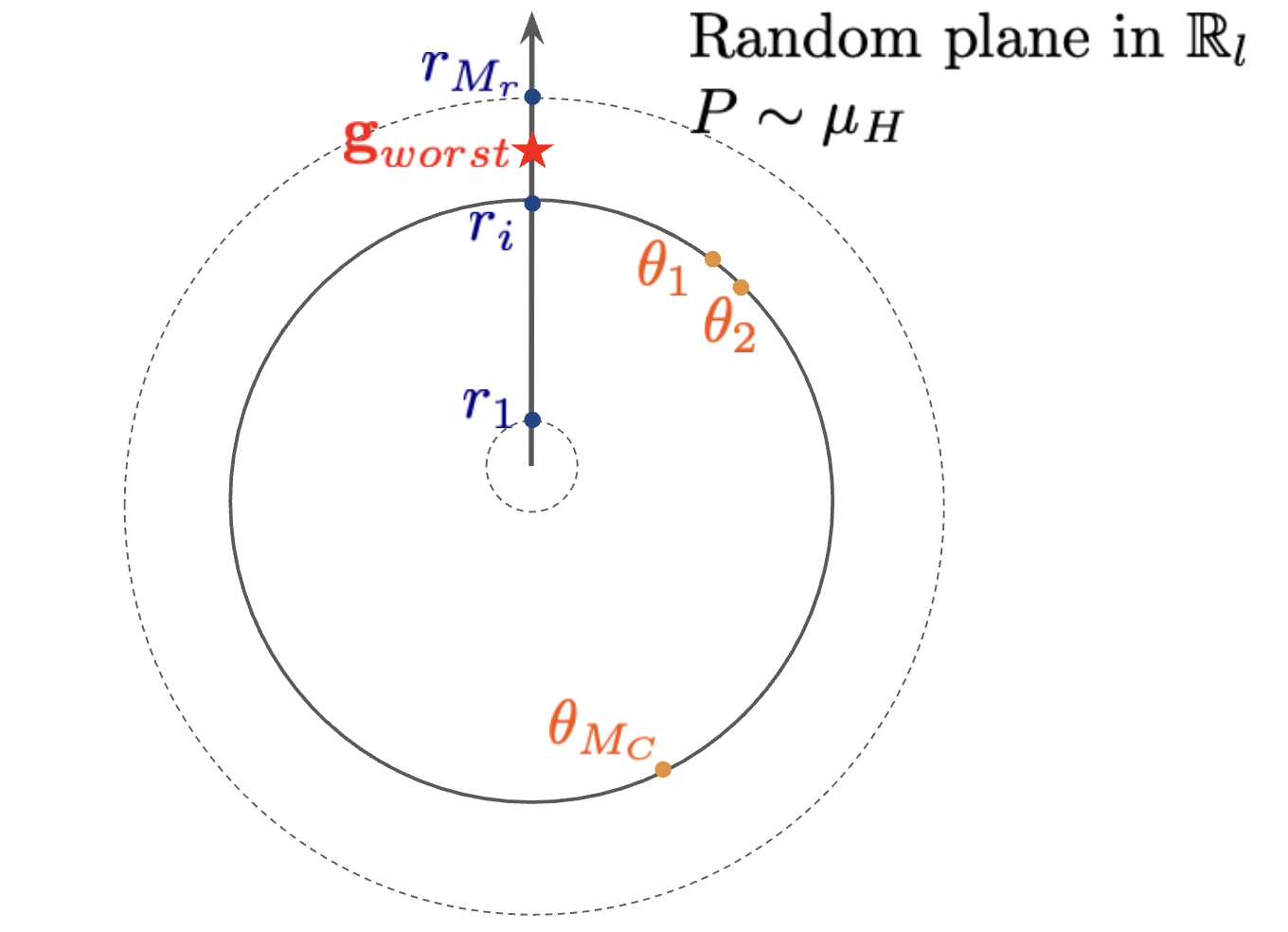}
    \caption{Slicing and dicing the sphere}
    \label{fig:slicingdicing}
\end{figure}

\subsection{Slicing and dicing the sphere}

First, we go from a problem on the entire $l$-dimensional sphere to one on a single slice, that is, a random $2D$ plane. That is, we reduce our problem to \textbf{bivariate interpolation on a random plane containing the origin and $\g_{worst}$ suffices}. The core observation is that a random plane is `good for $\mathcal{A}$' with high probability, meaning that $\mathcal{A}$ is likely to succeed on inputs from $\calN_l$ restricted to such a plane.

To see this more formally, notice that a point sampled from $\calN_l$ may equivalently be sampled as follows:
\begin{enumerate}
\item Sample a plane $P$ that contains the origin and $\g_{worst}$ from $\mu_H$, the Haar-random distribution over planes containing these two points.\footnote{To sample from $\mu_H$, we may draw a unit vector $q$ uniformly on $S^{l-1}$ and set $P=\text{span}(\g_{worst}, q).$ Sampling $q$ isotropically in the above fashion makes the construction invariant under any rotation that fixes $\g_{worst}$. That symmetry is exactly the one that permutes the set of planes containing $\g_{worst}$.} 
\item Sample from the marginal distribution of $\calN_l$ on $P$, which we call $\rho_P$.
\end{enumerate}  
 We may thus re-write the success guarantee of $\calA$, \Cref{eq:guarantee_A} as
\begin{equation}\label{eq:assumption}
    \operatorname{Pr}_{\x \sim \mathcal{N}_l}\left[ \calA\text{ succeeds}\right] = \operatorname{Pr}_{P \sim \mu_H}[\operatorname{Pr}_{\x \sim \rho_P}[\calA(x) \text{ succeeds}]] \geq 1-\delta.
\end{equation}
Thus $\calA$'s success guarantees on $P$ are:
\begin{lemma}[Sampling a good plane with high probability]\label{lemma:2}
Fix $\g_{\text{worst}} \in \mathbb{R}^l$,and let $P \sim \mu_H$ as defined above. Then for any $\delta_1>\delta$, 
\begin{equation}\label{eq:guarantee_outer}
    \operatorname{Pr}_{P \sim \mu_H}[\operatorname{Pr}_{\x \sim \rho_P}[\calA(x) \text{ succeeds}]\geq 1-\delta_1] \geq  1-\delta/\delta_1
\end{equation}
\end{lemma}
\begin{proof}
    Let $R:= \{P:\Pr_{x\sim \rho_P} [\calA(x) \text{ succeeds}] \geq 1-\delta_1\}$ and let
    \begin{equation}\label{eq:assumption_2}
    \operatorname{Pr}_{P \sim \mu_H}[P\in R] = k
\end{equation}
Then
\begin{align}
     1-\delta &\leq \operatorname{Pr}_{\x \sim \mathcal{N}_l}\left[ \calA\text{ succeeds}\right] \\
     &= \int_{P: \, P\in R} \Pr_{x\sim \rho_P}[\calA(x) \text{ succeeds}] \,\, d\mu_H(P)+ \int_{P: \, P\not\in R} \Pr_{x\sim \rho_P}[\calA(x) \text{ succeeds}] d\mu_H(P)\\
     &\leq k \max_{P\in R} \Pr_{x\sim \rho_P}[\calA(x) \text{ succeeds}] + (1-k) \max_{P\not\in R} \Pr_{x\sim \rho_P}[\calA(x) \text{ succeeds}]\\
     &\leq k + (1-k) (1-\delta_1) = 1-\delta_1+ \delta_1 k.
\end{align}

Here the first inequality follows by assumption (\Cref{eq:assumption}). Solving for $k$ yields the lemma.
\end{proof}

\Cref{lemma:2} says that if we sample a random $P$, then feed $\calA$ with points sampled from $\rho_P$, we would have high confidence that $\calA$ would evaluate correctly on those points. But the restriction of any polynomial to $P$ is a bivariate polynomial, because a plane only requires two parameters to describe fully. Moreover, the target point $\g_{worst}$ also lies in $P$. Thus, we could use $\calA$'s evaluations on points in $P$ as interpolation points for the bivariate polynomial $\calT_m$ restricted to $P$. 

We have successfully gone from a $l$-dimensional sphere to a  $2$-dimensional slice of it, but it is possible to go even further.  


\subsubsection{From a slice to a single ray}

We will further reduce the problem to univariate interpolation. Without loss of generality, we may assume that $\g_{\text{worst}}$ is on the $z$-axis of plane $P$'s coordinate system. (If not, we change coordinates; the polynomial in the new coordinates will have the same degree.) Instead of interpolating on $P$, we will actually interpolate on the $z$-axis. The polynomial we will interpolate is the restriction of $\calT_m$ to this axis, which is a univariate polynomial that we call $\calT_Z:[0,\infty)_Z\rightarrow \mathbb{R}$. By doing a more intricate version of the argument above, we will show how to obtain evaluations of $\calT_m$ on the $z$-axis, i.e. $\calT_Z(r_i)$ where $r_i\in [0,\infty)_Z.$

What will not work to obtain interpolation points is to run $\calA$ directly at $r_i$. We have no handle on how $\calA$ performs on the $z$-axis, as this is a worst-case collection of points -- recall that what we are calling the $z$-axis is simply the ray out from the origin to $\g_{worst}$. However, just as with the $l$-dimensional Gaussian $\calN_l$ itself, we may observe that $\calN_l$'s marginal distribution on $P$, $\rho_P$, further decomposes into a radial component and an angular component; indeed to sample a point $(r,\theta)$ from $\calN_l$ one can equivalently
\begin{enumerate}
    \item Sample its radial coordinate $r \sim \rho_R$, where $\rho_{R}$ is also the radial distribution of $\mathcal{N}(0,\sigma^2\mathbb{I}_l)$, which is a scaled $\chi(l)$-distribution: 
    \begin{equation}\label{eq:radial}
\rho_R^{(P)}(r)=\frac{1}{2^{l / 2-1} \Gamma(l / 2) \sigma^{l}} r^{l-1} \exp \left(-\frac{r^2}{2 \sigma^2}\right), \quad r \geq 0 .
\end{equation}
Equivalently, this step samples a random circumference of radius $r$ on the plane $P$. We denote such a circumference by $C_r$. 
    \item Sample a point in $C_r$ by sampling its angular coordinate $\theta\in (0, \pi],$ from the angular distribution $\rho_{\Theta},$ defined as\footnote{For intuition on why the angular distribution is given by \Cref{eq:polar}, we may envision the three-dimensional analog of this procedure: sample a random plane, $P$, passing through the z-axis, by uniformly sampling a random azimuthal angle $\varphi$ (angle from the $+Y$ axis) in $[0,2\pi)$. Then, sample $(r,\theta)\in P$ by sampling a radius from $\rho_R(r) = \frac{4}{\sigma^3\sqrt{\pi}} r^2 e^{-r^2/\sigma^2}$, and the polar angle $\theta$ from $\rho_{\Theta}(\theta)\propto \sin(\theta)$. Notice that the density of $\rho_{\Theta}$ peaks near $\theta = \pi/2$, i.e. points at radius $r$ near the equator ($\theta=\pi/2$) have a higher probability of being sampled than points at radius $r$ nearer to the poles ($\theta=0$ and $\theta=\pi$). This is because, if we allow $\varphi$ to vary between $(0,2\pi)$ but track the path of the point $(r,\theta)$, it traces out a ring of radius $r \sin(\theta)$ centered on the point $r\cos(\theta)$ on the $+Z$ axis. But the ring of radius $r$ between $\theta$ and $\theta+d\theta$, $Rg(\theta)$, has a larger area at $\theta=\pi/2$ than that at $\theta=0$; in fact, at any $r$, $\rho_{\Theta}$ is uniquely determined by imposing the requirement
\begin{equation}
    \frac{\text{Probability of sampling }\theta_1}{\text{Probability of sampling }\theta_2} = \frac{\text{area of $Rg(\theta_1)$}}{\text{area of $Rg(\theta_2)$} }
\end{equation}
which ensures that we are sampling from the right marginal. But this is nothing but
\begin{equation}
\frac{2\pi r\sin(\theta_1) d\theta}{2\pi r\sin(\theta_2) d\theta} = \frac{\sin(\theta_1)}{\sin(\theta_2)} = \frac{\rho_{\Theta}(\theta_1)}{\rho_{\Theta}(\theta_2)}
\end{equation}
as claimed. Extending this to $l$ dimensions, the path traced out by the point $(r,\theta)$ when one allows $\mathbf\varphi$ to take on all values in $(0,2\pi)^{l-1}$ is simply the set of all points in $l-1$ dimensions that are distance $r\sin(\theta)$ away from the point $r\cos(\theta)$ on the $+Z$ axis; this is the surface area of a $l-1$ dimensional sphere which is proportional to $(r\sin(\theta))^{l-2}$. Consequently
we must have 
\begin{equation}
\frac{\rho_{\Theta}(\theta_1)}{\rho_{\Theta}(\theta_2)} = \frac{\sin(\theta_1)^{l-2}}{\sin(\theta_2)^{l-2}}
\end{equation}
as claimed.}
    \begin{align}\label{eq:polar}
    \rho_{\Theta}^{(P)}(\theta) &\propto \sin^{l-2}(\theta), \quad \theta \in [0,\pi).
\end{align}
\end{enumerate}

The upshot is that, on plane $P$, a random circumference of radius $r$ is likely to be `good for $\calA$' with high probability. Letting $\x(r,\theta) \in P$ denote the point whose angular coordinates are $r,\theta$, we write this more rigorously as the following lemma:
\begin{equation}\label{eq:assumption_2}
    \operatorname{Pr}_{\x \sim \rho_P}\left[ \calA(\x) \text{ succeeds}\right] = \operatorname{Pr}_{r \sim \rho_R^{(P)}}[\operatorname{Pr}_{\theta \sim \rho_{\Theta}^{(P)}}[\calA(\x(r,\theta)) \text{ succeeds}]] \geq 1-\delta,
\end{equation}
and a counting argument similar to that of \Cref{lemma:2} reveals that the success guarantee of $\calA$ on $\rho_P$ proven there, translates into success when $\calA$ is fed with points from a random circumference $C_{r_i}$ where $r_i \sim \rho^{(P)}_R$:
\begin{lemma}[Sampling a good radius with high probability]\label{lemma:3}
If $\operatorname{Pr}_{\x \sim \rho_P}[\calA(x) \text{ succeeds}]\geq 1-\delta_1$, then for any $\delta_2>\delta_1$, 
\begin{equation}\label{eq:guarantee_outer}
    \operatorname{Pr}_{r \sim \rho_{R}^{(P)}}[\operatorname{Pr}_{\theta \sim \rho_{\Theta}^{(P)}}[\calA(\x(r,\theta)) \text{ succeeds}]\geq 1-\delta_2] \geq  1-\delta_1/\delta_2.
\end{equation}
\end{lemma}
We omit the proof because it is virtually identical to that of \Cref{lemma:2}. This lemma suggests the following natural strategy: we may recover \(\mathcal{T}_Z(r_i)\) by restricting the degree-\(m\) polynomial \(\mathcal{T}_m\) to \(C_{r_i}\) and interpolating from \(\mathcal{A}\)’s evaluations on that circle. 

Parametrize \(C_{r_i}\) by the polar angle \(\theta\), so the restriction of the bivariate degree-\(m\) polynomial \(\mathcal{T}_C\) to \(C_{r_i}\) is the expression \(\mathcal{T}_C(\cos\theta,\sin\theta)\). Our target is the value at the “North pole’’ (by convention, \(\theta=0\)):
\[
\mathcal{T}_C(\cos 0,\sin 0)=\mathcal{T}_C(1,0).
\]
Using the identity \(\sin^{2}\theta=1-\cos^{2}\theta\), rewrite every even power \(\sin^{2k}\theta\) as \((1-\cos^{2}\theta)^k\). Consequently, each monomial \(\cos^{a}\theta\,\sin^{b}\theta\) with even \(b\) becomes a polynomial in \(\cos\theta\), whereas with odd \(b\) it factors as \(\sin\theta\) times a polynomial in \(\cos\theta\). Therefore there exist univariate polynomials \(A,B\) with \(\deg A,\deg B\le m\) such that
\[
\mathcal{T}_C(\cos\theta,\sin\theta)\;=\;A(\cos\theta)\;+\;\sin\theta\,B(\cos\theta).
\]
To eliminate the \(\sin\theta\)-dependence, average the expression with its reflection \(\theta\mapsto -\theta\):
\begin{equation}\label{eq:symmetrization}
\frac{\mathcal{T}_C(\cos\theta,\sin\theta)+\mathcal{T}_C(\cos\theta,-\sin\theta)}{2}
\;=\;A(\cos\theta),
\end{equation}
because the terms odd in \(\sin\theta\) cancel. Doing the change of variables $x = \cos(\theta)$, define the univariate polynomial \(\mathcal{F}_C:[-1,1]\to\mathbb{R}\) by
\begin{equation}\label{eq:FcTc}
\mathcal{F}_C(x)\;:=\;\frac{\mathcal{T}_C\!\bigl(x,\sqrt{1-x^2}\bigr)+\mathcal{T}_C\!\bigl(x,-\sqrt{1-x^2}\bigr)}{2}
\;=\;A(x).
\end{equation}
(The square roots disappear after the cancellation, so \(\mathcal{F}_C\) is indeed a degree-\(m\) polynomial in \(x\).) By construction,
\[
\mathcal{F}_C(1)\;=\;A(1)\;=\;\mathcal{T}_C(1,0),
\]
so evaluating the original restricted polynomial at the North pole is equivalent to evaluating the univariate degree-\(m\) polynomial \(\mathcal{F}_C\) at \(x=1\).

To estimate \(\mathcal{F}_C\), draw angles \(\theta_j\stackrel{\text{i.i.d.}}{\sim}\rho_{\Theta}\) and set
\begin{equation}\label{eq:XY}
X_j := \cos\theta_j,\qquad
Y_j := \tfrac{1}{2}\Bigl[\mathcal{A}\bigl(\g(r_i,\theta_j)\bigr)+\mathcal{A}\bigl(\g(r_i,-\theta_j)\bigr)\Bigr].
\end{equation}
By \eqref{eq:symmetrization}, \((X_j,Y_j)\) are noisy samples of \(\mathcal{F}_C\). Interpolating a univariate degree-\(m\) polynomial from these samples yields an estimate of \(\mathcal{F}_C(1)=\mathcal{T}_C(1,0) = \calT_Z(r_i)\).

To summarize, the entire worst-to-average-case reduction consists in the following three steps:

\begin{enumerate}
    \item \textbf{Sample a random plane $P$ containing $\g_{worst}$.}
    \item \textbf{Interpolation on the \(z\)-axis (\Cref{alg:1}):} The target polynomial is the restriction to $\calT_m$ to the $Z$-axis, the degree-$m$ univariate polynomial $\calT_Z$. To obtain an estimate of $\calT_Z(\g_{worst})$, we interpolate on the points $\{(\rr_i, y_i)\}_i$ where $\rr_i\in [0,\infty)_Z$, sampled as $\rr_i\sim \rho_R^{(P)}$; and $y_i \approx \calT_m(\rr_i)$.
    \item \textbf{Interpolation on a circumference of radius \(|\rr_i|\) (\Cref{alg:2}):} To obtain $y_i$ at a given point $\rr_i$, we interpolate the target polynomial $\calF_C$ based on the points $(X_j, Y_j)$ given in \Cref{eq:XY}. Here $X_j = \cos(\theta_j)$ where the $\theta_j$ are angular coordinates of points on the circumference $C(r)$; also, $\theta_j\sim \rho_{\Theta}^{(P)}$. 
\end{enumerate} 

\subsection{Analysis of worst-to-average-case reduction}
To establish notation, we lay out the exact steps of the worst-to-average-case reduction as \Cref{alg:1}.

\begin{algorithm}[!htbp]
\caption{Worst-to-average-case reduction\label{alg:1}}
\begin{algorithmic}[1]
\Require A worst-case input Hamiltonian $\mathbf{g}_{\text{worst}} \in \mathbb{R}^l$, an average-case solver $\mathcal{A}$ for the quantity $D(\g)$. We assume $\g_{worst}$ is on the $z$-axis. 
\Ensure An estimate for $D(\g_{worst})$ by interpolating the polynomial $\calT_Z$. 
\State Sample a random plane $P$ containing the origin and $\g_{worst}.$
\State Sample $M_r$ points on the $z$-axis of $P$ from the radial distribution of $\calN_l$, denoted as $\rho_R$ (\Cref{eq:radial}).
\[
 \mathcal{S} = \{r_1, \dots, r_{M_r}\}\in [0,\infty)
\]
\State Construct a set of $\Delta$-separated points $\mathcal{S}_{\Delta} \subseteq \mathcal{S}$ using \Cref{proc:delta_separated}. 
\State For each $r_i\in\mathcal{S}_{\Delta}$, run \Cref{alg:2} to interpolate on the circumference $C(r_i)$ with failure probability $\delta/M_r$, obtaining the output $y_i$. 
\State Using the points $\bigl\{(r_i,y_i)\bigr\}_{i\in \mathcal{S}_{\Delta}},$ run Robust Efficient Berlekamp–Welch (\Cref{thm:REBW}) to interpolate a univariate polynomial $q_Z:[0,\infty)\rightarrow \mathbb{R}$ of degree $m$.
\State Output $q_Z(\mathbf{g}_{worst})$.
\end{algorithmic}
\end{algorithm}

\begin{algorithm}
\caption{Interpolation on a circumference $C(R)$ on a fixed plane \label{alg:2}}
\textbf{Notation.} Let $C(R)$ be the circumference consisting of all points in plane $P$ that are radius $R$ from the origin, and let plane $P\subset\mathbb{R}^l$ be equipped with polar coordinates $(r,\theta)$.

\begin{algorithmic}[1]
\Require An average-case solver $\mathcal{A}$ for the quantity $D(\g)$, radius $R$, failure probability $\delta$.
\Ensure An estimate for $D\bigl(z\bigr)$, where $z= \g(R,0)$ is the “North Pole’’ of $C$, by interpolating the polynomial $\calT_C$.
\State Sample $M_C$ points on $C(R)$ by drawing angular coordinates $\theta_j \overset{i.i.d.}{\sim} \rho_{\Theta}$: 
\begin{equation}
    \mathcal{T} := \{\theta_{1},\ldots,\theta_{M_C}\} \in (0,\pi].
\end{equation}
\State For each $j\in [M_C]$ let
\begin{align}
    X_j &:= \cos(\theta_j)\\
    Y_j &:= \mathcal{A}(\g\!\bigl(R,\theta_{j}\bigr))+\mathcal{A}(\g\!\bigl(R,-\theta_{j}\bigr))
\end{align}
and let $\calT':=\{X_1,\ldots X_{M_C}\}.$
\State Construct a set of $\Delta$-separated points in $[-1,1]$, $\mathcal{T}_{\Delta} \subseteq \calT'$, using the procedure in \Cref{proc:delta_separated}. 
\State Using the pairs $\bigl\{(X_{j},Y_{j})\bigr\}_{j\in \mathcal{T}_{\Delta}},$ run Robust Efficient Berlekamp–Welch
        (\Cref{thm:REBW}) to interpolate a univariate polynomial $q_{C}:[-1,1]\rightarrow [0,1]$ of degree $m$.
\State Output $q_{C}(1)$.
\end{algorithmic}
\end{algorithm}

The reduction requires a procedure to construct a set of $\Delta$-separated points from a sample $\mathcal{S}$ of points in $[a,b]$, which we give in \Cref{proc:delta_separated}. 
\begin{procedure}[Constructing a $\Delta$-separated subset]\label{proc:delta_separated} For an interval $[a,b]\subset\mathbb{R}$, upon input a sample set $\mathcal{S}\subset [a,b]$, the following procedure outputs a $\Delta$-separated subset $\mathcal{S}'\subseteq\mathcal{S}$:
\smallskip
\begin{enumerate}
\item Construct a $\Delta$-separated family of intervals $\mathcal{B}=(B_1,\dots,B_m)$ with $B_i\subset [a,b]$.
\item For each $i$, set $S_i:=\mathcal S\cap B_i$ and, if $S_i\neq\varnothing$, take $s_i$ to be an arbitrary point in $S_i.$
\item Return $\mathcal{S}':=\{\,s_i:\ S_i\neq\varnothing\,\}$.
\end{enumerate}
\end{procedure}

\begin{lemma}[Guarantees on \Cref{proc:delta_separated}]\label{lem:delta_separated}
	Let $\mathcal{B}=(B_1,\dots,B_m)$ be a $\Delta$–separated set of bins. Denote by $p_{\min}\;=\;\min_{B\in\mathcal{B}}\;
		            \Pr_{X\sim D}\!\bigl[X\in B\bigr]$
	the smallest bin-mass under the distribution
	$D : [a,b]\to[0,1]$. If
	\[
		k \geq
		\frac{1}{p_{\min}}
		\ln\!\Bigl(\frac{1}{(1-c)\delta}\Bigr)
	\]
	samples are drawn independently from $D$, then with probability at least
	$1-\delta$, at least $cm$ of the bins contain at least one sample. 
\end{lemma}

\begin{proof}
    In \Cref{proc:delta_separated}, $|\mathcal{S}'|$ is exactly the number of occupied bins in $\mathcal{B}$.

    Let $k = |\mathcal{S}|$. We will now compute how large $k$ must be to ensure that $|\mathcal{S}'| > (1-c)m$, i.e. there are at most $cm$ unoccupied bins in $\mathcal{B}.$ With $k$ samples, the probability that any particular $B_i$ is unoccupied is at most $(1-p_{\min})^k.$ Define the random variable $X_i = \mathbb{I}(B_i \text{ is unoccupied.})$ The number of unoccupied $B_i$'s is then $X := \sum_{i=1}^m X_i$, and its expectation is
    \begin{equation}
        \mathbb{E}[X] \leq m(1-p_{\min})^k
    \end{equation}
    by a union bound. 
    For any $k$ such that $c \leq \frac{1}{\delta}(1-p_{\min})^k$,
    \begin{equation}
        \Pr[X > cm] \leq \Pr[X > \frac{m}{\delta}(1-p_{\min})^k] \leq \Pr[X > \frac{1}{\delta}\mathbb{E}[X]]\leq \delta
    \end{equation}
    by Markov's inequality. Noting that $\Pr[\text{at least $(1-c)m$ bins in $\mathcal{B}$ are occupied}] = 1-\Pr[X > cm]$ completes the proof.
\end{proof}

Throughout this entire section, we will often find ourselves needing to bound ratios of Gamma functions. Gautschi's inequality \cite{Gautschi1959GammaIncompleteGammaInequalities} will come in very handy here:
\begin{theorem}[Gautschi's inequality]\label{thm:Gautschi}
    Let $x$ be a positive real number, and let $s \in(0,1)$. Then,

$$
x^{1-s}<\frac{\Gamma(x+1)}{\Gamma(x+s)}<(x+1)^{1-s}
$$
\end{theorem}

\subsubsection{Analysis of \Cref{alg:2} (Interpolation on a circumference)}
Since we will be using \Cref{proc:delta_separated} to construct a $\Delta$-separated sample out of the sampled points, we first show how to construct a $\Delta$-separated set of intervals. 

In the following, recall that $X(\theta) = \cos(\theta)$; let $\Theta: [-1,1]\rightarrow [0,\pi)$ be the inverse map, $\Theta(X) = \cos^{-1}(X)$. It will be more convenient to work directly with the effective distribution over the $X(\theta)$ induced by drawing $\theta \sim  \sin^{l-2}(\theta)$, which we will call $D: [-1,1]\rightarrow [0,1].$ We have \(\sin\theta=\sqrt{1-X^2}\) and \(\dfrac{d\theta}{dX}=-\dfrac{1}{\sqrt{1-X^2}}\). Hence
\[
D(X)=\rho_\Theta(\cos^{-1} (X))\,\Bigl|\tfrac{d\theta}{dX}\Bigr|
\;\propto\; (1-X^2)^{\frac{l-2}{2}}\,(1-X^2)^{-1/2}
=(1-X^2)^{\frac{l-3}{2}}.
\]
Normalizing (for \(l>1\)) gives
\begin{equation}\label{eq:D}
D(X)=\frac{\Gamma\!\left(\tfrac{l}{2}\right)}{\sqrt{\pi}\,\Gamma\!\left(\tfrac{l-1}{2}\right)}\,(1-X^2)^{\frac{l-3}{2}} = C_l (1-X^2)^{\frac{l-3}{2}},\qquad -1\le X\le 1.
\end{equation}
In particular $D(X)dX = \rho_{\Theta}(\theta)d\theta$ for $X=\cos(\theta)$, and therefore for any event $E$, $\Pr_{\theta\sim \rho_{\Theta}}[\theta\in E] = \Pr_{X\sim D}[\Theta(X) \in E].$ By Gautschi's inequality \Cref{thm:Gautschi}, we are also able to bound the ratio of $\Gamma$ functions in \Cref{eq:D}, resulting in a bound for the prefactor $C_l$:
\begin{equation}\label{eq:boundgamma}
    \sqrt{\frac{l-2}{2 \pi}}<C_l<\sqrt{\frac{l}{2 \pi}}.
\end{equation}
\begin{lemma}[Constructing a $\Delta$-separated set of intervals]\label{lem:bins_circ}
Given $l>3$ and a desired number of bins $B = \Theta(l)$, and let the distribution $D$ be as given in \Cref{eq:D}. 

We may construct a $\Delta = \Theta(l^{-3/2})$-separated set of bins on the support of $D$ by partitioning the interval $[-r,r]$ into bins of width $\Delta$, i.e.
\[
B_k:=\bigl[-r+2k\Delta,\,-r+(2k+1)\Delta\bigr]\subset[-r,r],\quad k=0,1,\dots,B-1,
\]
for $r= \Theta\left(\frac{1}{\sqrt{l}}\right).$ Let $\mathcal{B}:=(B_0,\dots,B_{B-1})$. Then $p_{min} = \min_{B\in\mathcal{B}}\;
		            \Pr_{X\sim D}\!\bigl[X\in B\bigr] = \Theta(1/l).$
\end{lemma}

\begin{proof}
    Fix $r = \Theta(1/\sqrt{l})$ and $\Delta = l^{-3/2}$. The number of bins is $B \approx \frac{2r-\Delta}{2\Delta} = \Theta(r/\Delta)= \Theta(l)$ as stated.
    
    Since we have assumed $l>3,$ in this regime $D$ decreases with $|X|$ and attains its minimum value of $D(r) = C_l (1-r^2)^{\frac{l-3}{2}}$ at the endpoints $-r$, $r$. We then compute the minimum probability mass in any bin using
    \begin{equation}
    \left(1-\frac{1}{l}\right)^{\frac{l-3}{2}}=\exp \left(\frac{l-3}{2} \log \left(1-\frac{1}{l}\right)\right)=e^{-1 / 2}\left(1+\frac{5}{4 l}+O\left(l^{-2}\right)\right)
    \end{equation}
    and \Cref{eq:boundgamma} to be
    \begin{equation}
       D(r)\Delta \sim \sqrt{l} \frac{1}{l^{3/2}} = \frac{1}{l},
    \end{equation}
    as stated.
\end{proof}

\begin{lemma}[Properties of $\calT_{{\Delta}}$]\label{lem:delta_separated_andgood_c} 
Let $\epsilon, \delta' \in (0,1]$. Suppose \Cref{alg:2} is such that
\(
\Pr_{\theta \sim \rho_{\Theta}^{(P)}}[\g(r,\theta) \text{ is good}] > 1-\delta
\). Assume $\delta<\frac{1}{1000}$. Given the samples
\(
\mathcal{T}=\{\theta_j\}_{j=1}^{k}
\stackrel{\text{i.i.d.}}{\sim}\rho_{\Theta}
\)
where
\begin{equation}
    k = O\left(l\,\log\frac{10}{\delta'}\right),
\end{equation}
\Cref{proc:delta_separated}, using the bins $\calB$ from \Cref{lem:bins_circ}, returns a $\Delta = O(l^{-3/2})$-separated sample $\mathcal{T}_{\Delta} = (X_j = \cos(\theta_j))_j$ such that, with probability at least \(1-2\delta'\), 
\begin{itemize}
    \item $.9B \leq |\calT_{\Delta}| \leq B$ where $B=\Theta(l)$;
    \item At most $(.21+\eps)B$ points in $\mathcal{T}_{\Delta}$ correspond to bad $\theta_j.$
\end{itemize}
\end{lemma}

\begin{proof}
Let $B$ be the number of bins in $\calB.$ We first prove the lower bound on the size of $\calT_{\Delta}$ output by \Cref{proc:delta_separated}. By \Cref{lem:delta_separated} and \Cref{lem:bins_circ}, if 
\begin{equation}
    k = O\left(l\,\log\frac{10}{\delta'}\right)
\end{equation}
points are sampled initially from $D$, then except with probability at most $\delta'$, at least a fraction 0.9 the bins will be occupied. Denote the set of occupied bins as $\calB_O$. 
Since exactly one sample from each occupied bin in $\mathcal B$ is put into $\calT_{\Delta},$ we have just proven that
\begin{equation}\label{eq:mbounds}
0.9B \leq |\calB_O| =|\mathcal{T}_{\Delta}|  \leq B.
\end{equation}
Condition on this. 

Next, we upper-bound the number of bad points in $\mathcal{T}_{\Delta}$. The first step is to upper-bound the number of bad bins (bins that are exceptionally likely to produce a bad point). For a point $r\in B_i$, define the event \(E(X):=\mathbb{I}(X \text{ is not good})\). Define the bad-bin set within $\mathcal{B}$:
\[
\mathcal{B}_{bad}:=\Bigl\{\,B_i:\ \Pr_{X\sim D}\bigl(E(X)\wedge X\in B_i\bigr)\ \ge\ \frac{100\delta}{B} \Bigr\}
\]
Then 
\begin{equation}
|\mathcal{B}_{bad}|\cdot \frac{100\delta}{B}
\;\le\;\sum_{i=1}^{|\mathcal{B}|}\Pr(E\wedge X\in B_i)
\;=\;\Pr(E)\;\le\;\delta,
\end{equation}
hence we may upper bound the total number of bad bins as 
\begin{equation}\label{eq:badbin_number_1}
    |\mathcal{B}_{bad}|\le \frac{B}{100} = 0.01B.
\end{equation} 
Also, for any bin $\calB_i \in \calB_{good}$ that is a good bin,
\begin{align}\label{eq:goodbin_mass_1}
    D(E(X)|X\in B_i) &\leq \frac{100\delta/B}{p_{min}} = \frac{100\delta}{\Theta(1)} \leq 200\delta\leq 0.20.
\end{align}
Here we have used that $B = \Theta(l)$ and the assumption $\delta<1/1000$.

We may finally complete the argument to upper-bound the number of bad points in $\mathcal{T}_{\Delta}$. Recall that one point is output per occupied bin. Then, defining the random variables $\{X_i\}$ where
\begin{equation}
    X_i := \mathbb{I}(\text{bad point is output from }\calB_i),
\end{equation}
$X_i$ and $X_j$ are independent random variables conditioned on both $\calB_i$ and $\calB_j$ being occupied. By linearity of expectation, the expected number of bad points in $\calT_{\Delta}$ is then
\begin{align}
    \sum_{B_i \in \mathcal{B}_O} \mathbb{E}[X_i] &= \sum_{B_i \in \mathcal{B}_{good} \cap \mathcal{B}_O  }\mathbb{E}[X_i] + \sum_{B_i \in \mathcal{B}_{bad} \cap \mathcal{B}_O  }\mathbb{E}[X_i]\\
    &\leq \sum_{B_i \in \mathcal{B}_{good} \cap \mathcal{B}_O  }\Pr[\text{output bad point from }\calB_i] + |\calB_{bad}|\\
    &\leq  0.20 B + 0.01B = 0.21 B.
\end{align}
where in the last step, to bound the first term we have coarsely upper-bounded $|\mathcal{B}_{good} \cap \mathcal{B}_O|\leq B$ and used \Cref{eq:goodbin_mass_1} and \Cref{eq:badbin_number_1}.

Finally, by a Chernoff bound due to the independence of $X_i$ and $X_j$, the number of bad points in $\calT_{\Delta}$, which is also $\sum_{i\in \mathcal{B}_O} X_i$, concentrates around its mean; namely
\begin{equation}
    \Pr\left[\left|\sum_{i\in \mathcal{B}_O} X_i - \mathbb{E}[\sum_{i\in \mathcal{B}_O} X_i]\right| \geq \eps |\calB_O|\right] \leq 2\exp(-2\eps^2 |\calB_O|) \leq 2\exp(-1.8 B\eps^2),
\end{equation}
where the last inequality follows from \Cref{eq:mbounds}.
The right-hand-side is at most $\delta'$ as long as $B\geq \Theta(\frac{1}{\eps^2}\log\frac{2}{\delta'}),$ which it is. Thus, with probability at most $\delta'$, the number of bad points in $\calT_{\Delta}$ can be bounded as
\begin{equation}
    0.21 B+\eps |\calB_O| \leq (0.21+\eps) B.
\end{equation}
\end{proof}

\begin{theorem}[Guarantees for interpolation on a circumference (\Cref{alg:2})\label{thm:circumference_interpolation}]
Fix $r\in (0,1]$. Assume access to an average-case solver $\calA$ in the sense of \Cref{eq:assumption} and assume that $r$ is a $\delta$-good radius for $\calA$, i.e.
\begin{equation}\label{eq:goodradius}
    \Pr_{\theta \sim \rho_{\Theta}^{(P)}}[|\mathcal{A}(\mathbf{g}(r,\theta))-D(\mathbf{g}(r,\theta))| \leq \epsilon_{\mathcal{A}}] \geq 1-\delta.
\end{equation}
Further assume that $\delta< \frac{1}{2000}.$
Let $\hat{y} = q_C(1)$ be the output of \Cref{alg:2}. With probability at least $1-O(\delta')$, $\hat{y}$ satisfies
    \begin{equation}
        |\hat{y}- D(\g(r,0))| \leq 2^{\Theta(n\log n)}\eps_{\calA},
    \end{equation}
    and the number of samples needed is 
    \begin{equation}
    M_C = O\left(l\,\log\frac{1}{\delta'}\right).
\end{equation}
\end{theorem}
\begin{proof}
Let $\mathcal{A}$ be the average-case solver. We will define a ``good" point on $P$ (labelled by its polar coordinates $(r,\theta)$) via
\begin{equation}\label{eq:meaningofgood}
    \g(r, \theta) \text{ is ``good" iff }  |\mathcal{A}(\mathbf{g}(r,\theta)) - D(\mathbf{g}(r,\theta))|\leq \epsilon_{\mathcal{A}} \text{ and } |\mathcal{A}(\mathbf{g}(r,-\theta)) - D(\mathbf{g}(r,-\theta))|\leq \epsilon_{\mathcal{A}}.
\end{equation}
Since \Cref{alg:2} only takes as inputs points on the same circumference $C(r)$, and we have assumed that $r$ is a $\delta$-good radius, we have
\begin{equation}
    \Pr_{\theta \sim \rho_{\Theta}^{(P)}}[\g(r,\theta) \text{ is good}] > 1-2\delta
\end{equation}
by a union bound and two applications of \Cref{eq:goodradius}. We have assumed that $2\delta< 10^{-3}$, so \Cref{lem:delta_separated_andgood_c} applies. 

We now bound the error made by $\calA$ in estimating $\calF_C$ on a ``good" point. This boils down to bounding the error between $\calF_C$ and $D$ as follows: recalling that $X = \cos(\theta)$, and the definitions of $
\calF_C$ and $\calT_C$ \Cref{eq:FcTc}, we have
\begin{align}
    |\calF_C(\cos \theta) - D(\g(r,\theta))| &\leq \frac{1}{2} |\calT_C(\cos(\theta), \sin(\theta)) - D(\g(r,\theta))| + \frac{1}{2} |\calT_C(\cos(-\theta), \sin(-\theta)) - D(\g(r,-\theta))| \\
    &= \frac{1}{2} |\calT_m(\g(r,\theta)) - D(\g(r,\theta))| + \frac{1}{2} |\calT_m(\g(r,-\theta)) - D(\g(r,-\theta))|. \label{eq:123}
\end{align}
Each of these terms may be bounded by \Cref{thm:polyapproxguarantees} as $\eps_{\calA}/2$ because for our choice of polynomial degree $m=\Theta(n)$ we have
\begin{equation}\label{eq:boundpolyerror}
    \lVert H(\g(r,\theta))\rVert = \left\lVert \sum_i g_i P_i \right\rVert \leq \lVert \g(r,\theta) \rVert_1\leq \sqrt{l} \lVert \g(r,\theta) \rVert_2 \leq \sqrt{l}.
\end{equation}
The last inequality follows because $\lVert \g(r,\theta) \rVert_2\leq 1$ as $r\leq 1$. So, recalling that $l=\Theta(n)$, for $t=O(1)$ and $\eps_{\calA}=2^{-n}$ we have that
\begin{equation}\label{eq:polydegree_1}
    m \geq \Theta(\lVert H(\g)\rVert t + \log(1/\eps_{\calA}))
\end{equation}
as required by \Cref{thm:polyapproxguarantees}. Combining \Cref{eq:123} and \Cref{eq:meaningofgood} yields
\begin{equation}\label{eq:meaningofgood_2}
    |Y_j-\calF_C(X_j)| =|(\mathcal{A}(\mathbf{g}(r,\theta))+\mathcal{A}(\mathbf{g}(r,-\theta))) - (\calF_{C}(\cos(\theta)) + \calF_{C}(\cos(-\theta)))| \leq 2(\epsilon_{\mathcal{A}}+\epsilon_{\mathcal{A}}).
\end{equation}
We therefore define a faulty index $j$ within $\calT$ as corresponding to a pair $(X_j,Y_j)$ on which the error bound \Cref{eq:meaningofgood_2} has been exceeded.

Applying \Cref{lem:delta_separated_andgood_c}, with $M_C= O(l \log \frac{1}{\delta'})$ samples, with probability $\geq 1-O(\delta')$ Line 3 succeeds in identifying a subset of points in the original sample, $\calT_{\Delta}\subseteq \calT$, such that $0.9B<|\calT_{\Delta}|<B$ with $B=\Theta(l)$. These points are also pairwise $\Delta = l^{-3/2}$-separated and contain $k_{faulty}\leq 0.25 B$ faulty indices. Condition on this.

Recall that we wish to interpolate the degree $m = \Theta(n)$ polynomial $\calF_C$ based on $\bigl\{(X_{j},Y_{j})\bigr\}_{j\in \mathcal{T}_{\Delta}}$. Choosing $B>\frac{m}{0.4}$ guarantees that 
\begin{equation}\label{eq:rebw_condition}
    m<|\calT_{\Delta}-2k_{faulty} -1|.
    \end{equation}
    Recall from \Cref{sec:preliminaries}
 that $l$, the number of terms in $H$, is also $\Theta(n)$, so $B>m/0.4$ is compatible with our assumption in \Cref{lem:bins_circ} that $B=\Theta(l).$ \Cref{eq:rebw_condition} is the condition \Cref{thm:REBW} needs to guarantee that Robust Efficient Berlekamp-Welch returns a degree-$m$ polynomial $q_C$ such that
\begin{equation}
  \bigl|\calF_C(X_j)-q_C(X_j)\bigr|\le \left( \frac{10}{\Delta}\right)^{O(m)} \eps_{\mathcal{A}} = (10l ^{3/2})^l\eps_{\mathcal{A}} = 2^{\Theta(n\log n)}\eps_{\calA},
\end{equation}
for at least $\Theta(m)$ points $X_j$. Call this set of points $\mathcal{T}'$. 

We may finally apply the modified Remez inequality (\Cref{lem:ModifiedRemez}) on the points in $\mathcal{T}'$ to bound the error of $q_C$ at $1$, given by the value of the polynomial difference $\calF_C-q_C$ on the interval $[-1,1]$, as
\begin{align}
    |\calF_C(1)-q_C(1)| &\leq deg(\calF_C) \left(\frac{2}{\Delta}\right)^{deg(\calF_C)}\max_{X_j\in \mathcal{T}'}|\calF_C(X_j)-q_C(X_j)| \leq 2^{\Theta(n\log n)}\eps_{\calA}.
\end{align}
 Since we also had $|\calF_C(1) - D(1)|\leq \epsilon_{\mathcal{A}}$, the conclusion follows with probability $1-O(\delta')$.
\end{proof}

\subsubsection{Analysis of \Cref{alg:1} (Worst-to-average-case reduction)}

We first derive some facts about the radial distribution of $\calN_l$, 
\begin{equation}\label{eq:rho_R}
\rho_R(r)=\frac{1}{2^{l / 2-1} \Gamma(l / 2) \sigma^{l}} r^{l-1} \exp \left(-\frac{r^2}{2 \sigma^2}\right), \quad r \geq 0 .
\end{equation}

The first and second moments are:
\begin{align}
\mathbb{E}[R] &=\sigma\sqrt{2}\;\frac{\Gamma\!\left(\frac{l+1}{2}\right)}{\Gamma\!\left(\frac{l}{2}\right)}.\\
\mathbb{E}[R^{2}] &= l\,\sigma^{2}.
\end{align}

The variance is thus:
\[
\operatorname{Var}(R) \;=\; \mathbb{E}[R^{2}] - \mathbb{E}[R]^2
\;=\; \sigma^{2}\!\left[\,l - 2\left(\frac{\Gamma\!\left(\frac{l+1}{2}\right)}{\Gamma\!\left(\frac{l}{2}\right)}\right)^{\!2}\right].
\]

The term $\frac{\Gamma\!\left(z+\frac{1}{2}\right)}{\Gamma\!\left(z\right)}$ is bounded by
\[
\frac{\Gamma\!\left(z+\tfrac12\right)}{\Gamma(z)}
= z^{1/2}\!\left(1-\frac{1}{8z}+\frac{1}{128z^{2}}+\frac{5}{1024z^{3}}+O(z^{-4})\right).
\]
Squaring and substituting \(z=\tfrac{l}{2}\) gives
\[
\left(\frac{\Gamma\!\left(\frac{l+1}{2}\right)}{\Gamma\!\left(\frac{l}{2}\right)}\right)^{\!2}
= z\!\left(1-\frac{1}{4z}+\frac{1}{32z^{2}}+\frac{1}{128z^{3}}+O(z^{-4})\right)
= \frac{l}{2}-\frac{1}{4}+\frac{1}{16l}+\frac{1}{32l^{2}}+O(l^{-3}).
\]
Hence, with $\sigma^2=\frac{1}{l}$, 
\begin{equation}
    \operatorname{Var}(R)= \frac{1}{2l}-\frac{1}{8l^{2}} + O\left(\frac{1}{l^3}\right).
\end{equation}

\begin{lemma}[Constructing a $\Delta$-separated set of intervals for $\rho_R$]\label{lem:bins_rad}
Given $l>3$ and a desired number of bins $B=\Theta(l)$, and let the distribution $\rho_R$ be as given in \Cref{eq:rho_R}. 

We may construct a $\Delta = \frac{1}{2l^{3/2}}$-separated set of bins on the support of $\rho_R$ by partitioning the interval $[1-\frac{1}{\sqrt{l}},1]$ into bins of width $\Delta$, i.e.
\[
B_k:=\left[1-\frac{1}{\sqrt{l}}+2k\Delta,\, 1-\frac{1}{\sqrt{l}}+(2k+1)\Delta\right]\subset \left[1-\frac{1}{\sqrt{l}},1\right],\quad k=0,1,\dots,B-1.
\]
 Let $\mathcal{B}:=(B_0,\dots,B_{B-1})$. Then $p_{min} = \min_{B\in\mathcal{B}}\;
		            \Pr_{r\sim \rho_R}\!\bigl[r\in \calB\bigr] = \frac{1}{2B\sqrt{\pi e^2}}.$
\end{lemma}

\begin{proof}
Take
\begin{align}
    a := 1-\frac{1}{\sqrt{l}}, \quad
    b := 1.
\end{align}
    Since $\rho_R$ has a peak at $\sigma\sqrt{l-1} = 1-\frac{1}{2l} + O(1/l^2)$ and is monotonically increasing or decreasing on either side of the peak, 
\begin{equation}
    \min_i \rho_R(B_i) = \min\{\rho_R(a)\Delta, \rho_R(b)\Delta\}.
\end{equation}
Note that 
$\rho_R(b) = \rho_R(1) = \frac{l^{l / 2} e^{-l / 2}}{2^{l / 2-1} \Gamma\left(\frac{l}{2}\right)}=\sqrt{\frac{l}{\pi}}\left(1-\frac{1}{6 l}+O\left(l^{-2}\right)\right)$ while to compute $\rho_R(a)$ at $a=1-\frac{1}{\sqrt{l}}$ we Taylor expand the log-density function:
\begin{align}
    \log\left(\frac{\rho_R\Big(1-\frac{c}{\sqrt{l}}\Big)}{\rho_R(1)}\right) &= (l-1)\log \Big(1-\frac{c}{\sqrt{l}}\Big) - \frac{l}{2}\left[\Big(1-\frac{c}{\sqrt{l}}\Big)^2-1\right]\\
    &= -c^2 + \frac{c-\frac{c^3}{3}}{\sqrt{l}} + \frac{\frac{c^2}{2}-\frac{c^4}{4}}{l} + O(l^{-3/2}).
\end{align}
Exponentiating this expression we obtain 
\begin{align}
    \frac{\rho_R\Big(1-\frac{c}{\sqrt{l}}\Big)}{\rho_R(1)} &= \exp\left[ -c^2 + \frac{c-\frac{c^3}{3}}{\sqrt{l}} + \frac{\frac{c^2}{2}-\frac{c^4}{4}}{l} + O(l^{-3/2})\right]\\
    &= e^{-c^2} \exp\left( \frac{c-\frac{c^3}{3}}{\sqrt{l}} + \frac{\frac{c^2}{2}-\frac{c^4}{4}}{l} + O(l^{-3/2})\right)\\
    &= e^{-c^2} \left(1+ \frac{c-\frac{c^3}{3}}{\sqrt{l}} + \frac{\frac{c^2}{2}-\frac{c^4}{4}}{l} + O(l^{-3/2})\right)
\end{align}
so that by setting $c=1$ and keeping the highest-order term, $\rho_R\left(1-\frac{1}{\sqrt{l}}\right)\approx \sqrt{\frac{l}{\pi}}e^{-1} (1+O(1/\sqrt{l})).$

This then allows us to compute the minimum bin mass:
\begin{equation}\label{eq:pmin}
    p_{min}= \rho_R(a)\frac{1}{2B\sqrt{l}} \approx \sqrt{\frac{l}{\pi}}e^{-1} \frac{1}{2B\sqrt{l}} = \frac{1}{2B\sqrt{\pi e^2}}.
\end{equation}
\end{proof}

In the following lemma, recall the definition of a ``good" radius $r$ as given in \Cref{eq:goodradius}.

\begin{lemma}[Properties of $\calS_{\Delta}$]\label{lem:delta_separated_andgood_radii} 
Suppose
\(
\Pr_{r\sim\rho_R}[r \text{ is good}]= 1-\delta
\) and assume $\delta<\frac{1}{1000}$. Given the samples
\(
\mathcal{S}=\{r_j\}_{j}
\stackrel{\text{i.i.d.}}{\sim}\rho_R
\) where
\[
 |\mathcal{S}| = O\left(l\,\log\frac{10}{\delta'}\right)
\]
\Cref{proc:delta_separated} returns a \(\Delta = O(l^{-3/2})\)-separated sample $\mathcal{S}_{\Delta}\subseteq \mathcal{S}$, such that, with probability at least \(1-2\delta',\)
\begin{itemize}
    \item $.9B \leq |\calS_{\Delta}| \leq B$ where $B=\Theta(l)$;
    \item $\max_{r_i \in \mathcal{S}_{\Delta}} \leq 1$;
    \item $\calS_{\Delta}$ contains at most $.11B$ indices $i$ corresponding to $r_i$ that are not good. 
\end{itemize}
\end{lemma}

\begin{proof}
The proof follows similarly to that of \Cref{lem:delta_separated_andgood_c}: We first prove the lower bound on the size of $\calS_{\Delta}$ output by \Cref{proc:delta_separated}.
Let $B = \Theta(l)$ be the number of bins in $\calB.$ Since $p_{min} =\Theta(\frac{1}{l})$ (c.f. \Cref{lem:bins_rad} ), we may argue exactly identically that 
\begin{equation}
    k = O\left(l\,\log\frac{10}{\delta'}\right)
\end{equation}
points are sufficient to guarantee that except with probability $\delta'$, 
\begin{equation}\label{eq:mbounds_rad}
0.9B \leq |\calB_O| =|\mathcal{S}_{\Delta}|  \leq B.
\end{equation}
Condition on this.

Next, we upper-bound the number of bad points (a point such that $r$ is not good) in $\mathcal{S}_{\Delta}$. The first step is to upper-bound the number of bad bins (bins that are exceptionally likely to produce a bad point). For a point $r\in B_i$, define the event \(E(r):=\mathbb{I}(r \text{ is not good})\). Define the bad-bin set within $\mathcal{B}$:
\[
\mathcal{B}_{bad}:=\Bigl\{\,B_i:\ \Pr_{r\sim \rho_R}\bigl(E(r)\wedge r\in B_i\bigr)\ \ge\ \frac{100\delta}{B} \Bigr\}
\]
Then 
\begin{equation}
|\mathcal{B}_{bad}|\cdot \frac{100\delta}{B}
\;\le\;\sum_{i=1}^{|\mathcal{B}|}\Pr(E\wedge r\in B_i)
\;=\;\Pr(E)\;\le\;\delta,
\end{equation}
hence we may upper bound the total number of bad bins as 
\begin{equation}\label{eq:badbin_number}
    |\mathcal{B}_{bad}|\le \frac{B}{100} = 0.01B.
\end{equation} 
Also, for any bin $\calB_i \in \calB_{good}$ that is a good bin,
\begin{align}\label{eq:goodbin_mass}
    \rho_R(E(r)|r\in B_i) &\leq \frac{100\delta/B}{\rho_R(B_i)}\leq \frac{100\delta/B}{p_{min}} = 200\delta \sqrt{\pi} e \leq 0.096.
\end{align}

We now turn to upper-bounding the number of bad points in $\mathcal{S}_{\Delta}$ created via our above sampling procedure. Recall that one point is output per occupied bin. Defining the random variables $\{X_i\}$ where
\begin{equation}
    X_i := \mathbb{I}(\text{bad point is output from }\calB_i),
\end{equation}
$X_i$ and $X_j$ are independent random variables conditioned on both $\calB_i$ and $\calB_j$ being occupied. By linearity of expectation, the expected number of bad points in $\mathcal{S}_{\Delta}$ is then
\begin{align}
    \sum_{B_i \in \mathcal{B}_O} \mathbb{E}[r_i] &= \sum_{B_i \in \mathcal{B}_{good} \cap \mathcal{B}_O  }\mathbb{E}[r_i] + \sum_{B_i \in \mathcal{B}_{bad} \cap \mathcal{B}_O  }\mathbb{E}[r_i]\\
    &\leq \sum_{B_i \in \mathcal{B}_{good} \cap \mathcal{B}_O  }\Pr[\text{output bad point from }\calB_i] + |\calB_{bad}|\\
    &\leq  0.096 B + 0.01B = 0.106B.
\end{align}
where in the last step, to bound the first term we have coarsely upper-bounded $|\mathcal{B}_{good} \cap \mathcal{B}_O|\leq |\calB|=B$ and used \Cref{eq:goodbin_mass} and \Cref{eq:badbin_number}.

Finally, by a Chernoff bound due to the independence of $r_i$ and $r_j$, the number of bad points in $\calS_{\Delta}$, which is also $\sum_{i\in \mathcal{B}_O} r_i$, concentrates around its mean; that is
\begin{equation}
    \Pr\left[\left|\sum_{i\in \mathcal{B}_O} r_i - \mathbb{E}\left[\sum_{i\in \mathcal{B}_O} r_i\right]\right| \geq \eps |\calB_O| \right] \leq 2\exp(-2\eps^2 |\calB_O|) \leq 2\exp(-1.8 B\eps^2),
\end{equation}
where the last inequality follows from \Cref{eq:mbounds_rad}. The right-hand-side is at most $\delta'$ as long as $B\geq \Theta(\frac{1}{\eps^2}\log\frac{2}{\delta'}).$ Choosing $\eps=.04$, this is easily guaranteed as we have already assumed that $B = \Theta(l).$ Thus, with probability at most $\delta'$, the number of bad points in $\calS_{\Delta}$ can be bounded as
\begin{equation}
    |\{\text{Bad points in }\calS_{\Delta}\}| \leq .106B + \epsilon B \leq (0.11)B.
\end{equation}
\end{proof}

Finally, we put everything together.
\begin{theorem}[Worst-to-average-case reduction\label{thm:outer}]
Let $\hat{p} = q_Z(\g_{worst})$ be the output of \Cref{alg:1}. Assume access to an average-case solver $\calA$ such that 
\begin{equation}\label{eq:Asucceess_2}
    \Pr_{\g\sim \calN_l}[\calA(\g) \text{ succeeds}] \geq 1-\delta
\end{equation}
where $\delta<10^{-12}$. 
Then, with probability at least $1-O(\delta^{1/4})$, $\hat{p}$ satisfies
    \begin{equation}
        |\hat{p}- D(\g_{worst})| \leq 2^{\Theta(n\log n)}\epsilon_{\mathcal{A}}.
    \end{equation}
    Moreover, \Cref{alg:1} runs in time polynomial in $n$.
\end{theorem}
It is certainly possible to optimize the constant $\delta$; we have made no attempt to do so.
\begin{proof}
By \Cref{lemma:2}, with probability $\geq 1-\sqrt\delta$, the plane $P$ sampled in Step 1 is good for $\calA$, i.e. 
\begin{equation}
    \operatorname{Pr}_{\x \sim \rho_P}[\calA(x) \text{ succeeds}]\geq 1-\sqrt \delta.
\end{equation}
Condition on having sampled a good plane. Then, \Cref{lemma:3} gives guarantees on the likelihood (over the good plane) of sampling a good radius: with probability at least $1-\delta^{1/4}$, each radius $r_i\sim \rho_R^{(P)}$ sampled in Step 2 is $\delta^{1/4}$-good for $\calA$, satisfying
\begin{equation}
    \operatorname{Pr}_{\theta \sim \rho_{\Theta}^{(P)}}[\calA(\x(r_i,\theta)) \text{ succeeds}]\geq 1-\delta^{1/4}.
\end{equation}
Plugging the choice $\delta\leftarrow \delta^{1/4}<10^{-3}$ into \Cref{lem:delta_separated_andgood_radii}, then, and also choosing $\delta'=\delta$, gives guarantees on the $\Delta$-separated set $\calS_{\Delta}\subseteq \calS$ constructed in line 3 of \Cref{alg:1}: Recalling that $l=\Theta(n)$, as long as $M_r= |\calS| =O\left(n \log\left(\frac{1}{\delta}\right)\right)$, then with probability $1-O(\delta^{1/4})$, we have $0.9B<|\calS_{\Delta}|<B$ with $B=\Theta(l)$. These points are also pairwise $\Delta = l^{-3/2}$-separated and contain at most $0.11 B$ bad radii. We condition on this.

\Cref{thm:circumference_interpolation} with $\delta\leftarrow \delta^{1/4}$ and $\delta'\leftarrow \delta^{1/4}/B$ gives guarantees on the output of  \Cref{alg:2} at a $\delta^{1/4}$-good radius: its output $y_i$ satisfies
\begin{equation}\label{eq:goodr}
    \Pr[|y_i- D(\g(r_i))| \leq 2^{\Theta(n\log n)}\eps_{\calA}] \geq 1-O(\delta^{1/4}/B).
\end{equation}
Since there are $|\calS_{\Delta}| > 0.9B$ radii at which we run \Cref{alg:2}, by a union bound, with probability at least $1-O(\delta^{1/4})$, all $y_i$ corresponding to good radii are $2^{\Theta(n\log n)}\eps_{\calA}$-accurate estimates of $D(\g(r_i))$. Now, noting that $r_i<1$  for all $r_i\in \calS_{\Delta}$ (point 2 of \Cref{lem:delta_separated_andgood_radii}),  \Cref{thm:polyapproxguarantees}
 guarantees
 \begin{equation}
    |\calT_Z(r_i)-D(\g(r_i))|\leq \eps_{\calA}
\end{equation}
by a similar line of reasoning as \Cref{eq:boundpolyerror}. That is, we have conditioned on the fact that at most $.11B$ pairs $(r_i, y_i) \in \calS_{\Delta}$ are faulty evaluations, i.e.
\begin{equation}
        |y_i- \calT_Z(r_i)| > 2^{\Theta(n\log n)}\eps_{\calA}.
\end{equation}
Let us denote the number of faulty evaluations as $k_{faulty}.$

In Line 5, we interpolate the degree-$m$ polynomial $\calT_Z$ based on the possibly faulty evaluations $\bigl\{(r_i,y_{i})\bigr\}_{j\in \mathcal{S}_{\Delta}}$. Since $\calT_Z$ is a degree $m$ polynomial, choosing $B=2m$ guarantees that 
\begin{equation}\label{eq:rebw_condition-2}
    m<|\calS_{\Delta}|-2k_{faulty} -1.
    \end{equation}
Again, by our choice of $m=\Theta(n)$, $B=2m$ is compatible with our earlier assumption in \Cref{lem:bins_rad} that $B=\Theta(l).$ \Cref{eq:rebw_condition-2} is the condition \Cref{thm:REBW} needs to guarantee that Robust Efficient Berlekamp-Welch returns a degree-$m$ polynomial $q_Z$ such that
\begin{equation}
    \left|\calT_Z(r_i)-q_Z(r_i)\right| \leq \left(\frac{10}{l^{-3/2}} \right)^{2|B|}2^{\Theta(n\log n)}\eps_{\calA} =2^{\Theta(n\log n)}\eps_{\calA}
\end{equation}
for at least $m+1$ points $i$. Let us call this set of points $\mathcal{S}'.$ 

Noting that our worst-case point $\g_{worst}$, being a length-$\Theta(n)$ vector of Ising coefficients of magnitude $\Theta(1)$, is of radial distance $\lVert \g_{worst}\rVert_2=\Theta(\sqrt{n})$ from the origin, we apply the discrete Remez inequality \Cref{lem:Remez} on the points in $\calS'$ to bound the error of $q_Z(\mathbf{g}_{worst})$ at $\lVert \g_{worst}\rVert_2=\sqrt{n}$, by bounding the polynomial $err:=\calT_Z - q_Z$ at $\mathbf{g}_{worst}$ as
\begin{align}
    |\calT_Z(\mathbf{g}_{worst})-q_Z(\mathbf{g}_{worst})| &\leq (e^2(\Delta (2(m+1)))^{-1} \sqrt{n})^{m+1}\max_{r_j\in \mathcal{S}'}|\calT_Z(r_j)-q_Z(r_j)| \\
    &= 2^{\Theta(n\log n)}\eps_{\calA}
\end{align}
as $\Delta=l^{-3/2}$ and $l,m=\Theta(n).$ Finally, using \Cref{thm:polyapproxguarantees} again with $t=O(1)$, $\eps_{\calA}=2^{-n}$ and 
\begin{equation}\label{eq:boundpolyerror}
    \lVert H(\mathbf{g}_{worst})\rVert = \left\lVert \sum_i g_i P_i \right\rVert \leq \lVert \mathbf{g}_{worst} \rVert_1\leq \sqrt{l} \lVert \mathbf{g}_{worst} \rVert_2 \leq \Theta(n),
\end{equation}
we again have that
 $|\calT_Z(\mathbf{g}_{worst}) - D(\mathbf{g}_{worst})|\leq \eps_{\calA}$. Thus, $\hat{p}=q_Z(\mathbf{g}_{worst})$has an overall error bounded by
 \begin{equation}
     |q_Z(\mathbf{g}_{worst}) - D(\mathbf{g}_{worst})| \leq 2^{\Theta(n\log n)}\eps_{\calA}.
 \end{equation}
 
 The overall polynomial runtime is given by the fact that there are $M_r = O(n \log \frac{1}{\delta})$ runs of \Cref{alg:2} (\Cref{lem:delta_separated_andgood_radii}) and each run takes time polynomial in its $M_C = O(n \log \frac{n}{\delta})$ samples.
\end{proof}

\section{Acknowledgments}
We thank Allen Liu and Chen Lu for teaching us a significant amount of high-dimensional geometry. We also thank Shaun Datta, Isaac Kim, Jonas Haferkamp, Nick Hunter-Jones, Adam Bouland, Dominik Hangleiter, Umesh Vazirani, Urmila Mahadev, Soumik Ghosh and Bill Fefferman for helpful discussions. This material is based upon work supported by the U.S. Department of Energy, Office of Science, National Quantum Information Science Research Centers, Quantum Systems Accelerator. This work was done in part while the author was visiting the Simons Institute for the Theory of Computing.

\bibliographystyle{alpha}  
\bibliography{main}  
\end{document}